\documentclass[a4paper,twocolumn,11pt,accepted=2021-11-08]{quantumarticle}
\pdfoutput=1
\usepackage{graphicx}
\usepackage{epsfig}
\usepackage{epstopdf}
\usepackage[usenames,dvipsnames]{xcolor}
\usepackage{braket}
\usepackage{amsthm, mathtools}
\usepackage{amssymb}
\usepackage{mdframed}
\usepackage{bm}
\usepackage{enumitem}
\usepackage{dsfont}
\usepackage{bbold}
\usepackage{eurosym}
\usepackage{microtype}
\usepackage{hyperref}

\allowdisplaybreaks

\DeclareMathOperator{\poly}{poly}

\newcommand{\be}{\begin{equation}}
\newcommand{\ee}{\end{equation}}
\newcommand{\ba}{\begin{aligned}}
\newcommand{\ea}{\end{aligned}}

\newcommand{\bc}{\begin{center}}
\newcommand{\ec}{\end{center}}
\newcommand{\beq}{\begin{equation}}
\newcommand{\eeq}{\end{equation}}
\newcommand{\beqq}{\begin{equation*}}
\newcommand{\eeqq}{\end{equation*}}
\newcommand{\beqa}{\begin{align}}
\newcommand{\eeqa}{\end{align}}
\newcommand{\barr}{\begin{array}}
\newcommand{\earr}{\end{array}}
\newcommand{\bi}{\begin{itemize}}
\newcommand{\ei}{\end{itemize}}

\newcommand{\Tr}{\ensuremath{\,\mathrm{Tr}}}

\newtheorem{lem}{Lemma}

\newtheorem{theo}{Theorem}

\newtheorem{prot}{Protocol}
\newtheorem{protocol}{Protocol}

\begin{document}

\title{Efficient verification of Boson Sampling}

\author{Ulysse Chabaud}
\affiliation{Sorbonne Universit\'e, CNRS, LIP6, F-75005 Paris, France}
\email{ulysse.chabaud@gmail.com}
\orcid{0000-0003-0135-9819}
\author{Fr\'ed\'eric Grosshans}
\affiliation{Sorbonne Universit\'e, CNRS, LIP6, F-75005 Paris, France}
\orcid{0000-0001-8170-9668}
\author{Elham Kashefi}
\affiliation{Sorbonne Universit\'e, CNRS, LIP6, F-75005 Paris, France}
\affiliation{School of Informatics, University of Edinburgh, 10 Crichton Street, Edinburgh, EH8 9AB}
\author{Damian Markham}
\affiliation{Sorbonne Universit\'e, CNRS, LIP6, F-75005 Paris, France}
\affiliation{JFLI, CNRS, National Institute of Informatics, University of Tokyo, Tokyo, Japan}

%--------------------------------------------------------------------------------

\begin{abstract}

The demonstration of quantum speedup, also known as quantum computational supremacy, that is the ability of quantum computers to outperform dramatically their classical counterparts, is an important milestone in the field of quantum computing. While quantum speedup experiments are gradually escaping the regime of classical simulation, they still lack efficient verification protocols and rely on partial validation. Here we derive an efficient protocol for verifying with single-mode Gaussian measurements the output states of a large class of continuous-variable quantum circuits demonstrating quantum speedup, including Boson Sampling experiments, thus enabling a convincing demonstration of quantum speedup with photonic computing. 
Beyond the quantum speedup milestone, our results also enable the efficient and reliable certification of a large class of intractable continuous-variable multimode quantum states.
 
\end{abstract}

%--------------------------------------------------------------------------------

\maketitle

%--------------------------------------------------------------------------------

\section{Introduction}

\noindent Quantum information promises many technological applications beyond classical information~\cite{wiesner1983conjugate,bennett1992communication,bennett1984quantum,harrow2009quantum}. Shor's famous factoring algorithm~\cite{Shor1999} has highlighted the possibility of spectacular quantum advantages over classical computing. The ability of quantum computers to greatly outperform their classical counterparts is referred to as ``quantum computational supremacy''~\cite{preskill2012quantum,harrow2017quantum}, or quantum speedup~\cite{ronnow2014defining}. We use the latter terminology in this article.

While a universal quantum computer factoring large products of primes would provide a compelling evidence of quantum speedup, the experimental requirements associated with such a demonstration preclude its realisation in the near future~\cite{haner2016factoring}. 
For that reason, various sub-universal models of quantum computing, such as Boson Sampling~\cite{Aaronson2013}, have been introduced~\cite{terhal2002adaptive,shepherd2009temporally,bremner2011classical,boixo2018characterizing,mezher2019efficient}. Beyond their conceptual relevance, these models, although not believed to possess the full computational power of a universal quantum computer, enable the possibility of a demonstration of quantum speedup in the near term, under certain complexity-theoretic assumptions. 
Each of these sub-universal models is associated with a computational problem that the model solves efficiently, but which is hard to solve for classical computers. These computational problems are sampling problems, for which the task is to draw random samples from a target probability distribution fixed by the model. 
In practice, any experimental implementation would suffer from noise and imperfections and would sample from an imperfect probability distribution which approximates the target one. Hence, the fact that even an approximate version of the sampling problem associated with a sub-universal model is hard for classical computers is crucial for an experimental demonstration of quantum speedup. For all models, such an approximate sampling hardness, which corresponds to the fact that a probability distribution which has a small constant total variation distance with the target probability distribution is still hard to sample classically, comes at the cost of assuming some additional reasonable---though unproven---conjectures. 

The experimental demonstration of quantum speedup thus involves: (i) a quantum device solving efficiently a sampling task which is provably hard to solve for classical computers under reasonable theoretical assumptions, together with (ii) a verification that the quantum device indeed solved the hard task~\cite{harrow2017quantum}. While the former has been recently achieved with random superconducting circuits~\cite{arute2019quantum} and Gaussian Boson Sampling~\cite{zhong2020quantum}, the latter is still incomplete and relying on various questionable assumptions~\cite{aaronson2016complexity,dalzell2020many,barak2020spoofing}, and verification is a crucial missing point for a convincing demonstration of quantum speedup.

The setting of verification involves two parties, which in the language of interactive proof systems are referred to as the prover and the verifier. The prover is conceived as a powerful, untrusted party, while the verifier is a trusted party with restricted computational power. The verifier asks the prover to perform a computational task and verifies its correct behaviour, with possibly multiple rounds of interaction. While this setting corresponds to a cryptographic scenario between two parties, e.g., in the case of delegated computing, it also embeds the notion of certification of a physical experiment, where the experiment behaves as the prover and the experimenter as the verifier. 
Assumptions can be made on the prover, such as restricting its computational capabilities or assuming that it follows a specific behaviour. In the physical picture, this corresponds for example to choosing a noise model for an experiment. A common assumption is the so-called independently and identically distributed (i.i.d.) assumption, i.e., assuming that the output of the quantum experiment is the same at each run. Assuming i.i.d.\@ behaviour for the prover leads to more efficient verification protocols, and this assumption can usually be removed using cryptographic techniques at the cost of an increased number of runs in the protocol~\cite{renner2007symmetry,gheorghiu2017verification}.
We use the language of interactive proof systems in what follows, but we emphasise that our results are not restricted to that particular context and are relevant especially in experimental scenarios. There has been much work done on verification of universal quantum computation~\cite{broadbent2009universal,gheorghiu2017verification,mahadev2018classical}, much less when the verification is limited to sub-universal models. We now briefly review the state of affairs in this direction.

A verification protocol for a quantum speedup experiment should guarantee the closeness in total variation distance between the experimental probability distribution and the target probability distribution~\cite{harrow2017quantum}. Such a verification is especially difficult because of the nature of the computational task at hand, i.e., sampling from an anti-concentrating probability distribution over an exponentially large sample space.
In particular, any efficient non-interactive verification of current quantum speedup experiments with a verifier restricted to classical computations requires additional cryptographic assumptions~\cite{hangleiter2019sample}. Existing verification protocols with a classical verifier based on total variation distance either rely on little-studied assumptions~\cite{shepherd2009temporally,kahanamoku2019forging,yung2020anti}, or induce an overhead for the prover~\cite{mahadev2018classical,brakerski2018cryptographic,brakerski2020simpler} which prevents a near-term use for an experimental demonstration of quantum speedup. 

Weaker but more resource-efficient methods of verification with a classical verifier, which are sometimes referred to as validation~\cite{spagnolo2014experimental}, consist in performing a partial verification where only specific properties of the experimental probability distribution are tested, rather than closeness in total variation distance to the ideal distribution. For example, the verifier may perform statistical tests using experimental samples~\cite{boixo2018characterizing,arute2019quantum,drummond2021simulating}, compare the experimental probability distribution with specific mock-up distributions~\cite{aaronson2013bosonsampling}, or benchmark the individual components of the experimental setup. Ultimately, validation relies on making additional assumptions about the inner functioning of the quantum device~\cite{ferracin2019accrediting}.

In order to avoid relying on such undesirable assumptions, while keeping verification within near-term experimental reach, another way for performing verification is to allow the verifier to have minimal quantum capabilities. In that context, the purpose of the verifier is to obtain a reliable bound on the trace distance between the output state produced by the prover and the ideal target state, which in turn implies a bound on the distance between the tested and target probability distributions of the samples~\cite{NielsenChuang}. This can be done by obtaining, e.g., an estimate of the fidelity with the target state, or a fidelity witness, i.e., a tight lower bound on the fidelity~\cite{fuchs1999cryptographic}.

In the context of quantum computing in finite-dimensional Hilbert spaces~\cite{NielsenChuang}, this minimal quantum capability corresponds to being only able to prepare single-qubit or qudit states or to perform local measurements. Protocols for verification of Instantaneous Quantum Polynomial time circuits~\cite{shepherd2009temporally} with these minimal requirements have been derived under the i.i.d.\@ assumption with single-qubit states~\cite{mills2017information} or with local measurements~\cite{Hangleiter2016} and more recently without the i.i.d.\@ assumption with single-qubit states~\cite{kapourniotis2019nonadaptive} or with local measurements~\cite{takeuchi2018verification}.

Similarly in the context of infinite-dimensional Hilbert spaces~\cite{Braunstein2005}, which includes Boson Sampling and related sub-universal models~\cite{Aaronson2013,Lund2014,olson2015sampling,hamilton2017gaussian,chabaud2017continuous,chakhmakhchyan2017boson,lund2017exact}, this minimal quantum capability corresponds to being only able to prepare single-mode Gaussian states, or to perform single-mode Gaussian measurements. Gaussian state preparations and measurements are at the same time well understood theoretically~\cite{adesso2014continuous}, classically simulable~\cite{Bartlett2002}, and routinely implemented at large scales experimentally~\cite{ferraro2005gaussian,yokoyama2013ultra,yoshikawa2016invited}. However, there is no efficient verification protocol using single-mode Gaussian state preparation nor single-mode Gaussian measurements for Boson Sampling with input single photons: current methods used for the validation of Boson sampling are either not scalable or only provide partial certificates on the tested probability distribution~\cite{spagnolo2014experimental,opanchuk2018simulating,flamini2018photonic,agresti2019pattern,brod2019photonic,wang2019boson,walschaers2020signatures}.

%--------------------------------------------------------------------------------

\section{Results}

\noindent Our contribution is three-fold: we introduce three noninteractive protocols (detailed below) for the reliable certification of continuous-variable quantum states, where the verifier only needs to perform classical computations and single-mode Gaussian measurements, namely heterodyne detection~\cite{yuen1980optical}, in order to verify an unknown state sent by an untrusted prover. The three protocols share the same structure, depicted in Fig.~\ref{fig:setup}. 

For each of our three protocols, we derive one version assuming i.i.d.\@ behaviour for the prover and another version making no assumptions whatsoever on the prover. Other than that, all protocols rely solely on the verifier having access to single-mode heterodyne detection and on the validity of quantum mechanics. Moreover, all protocols are robust and efficient, as they provide analytical confidence intervals and require a polynomial number of measurements in the number of modes and in the size of the confidence interval.

\begin{figure}[t]
	\begin{center}
		\includegraphics[width=\columnwidth]{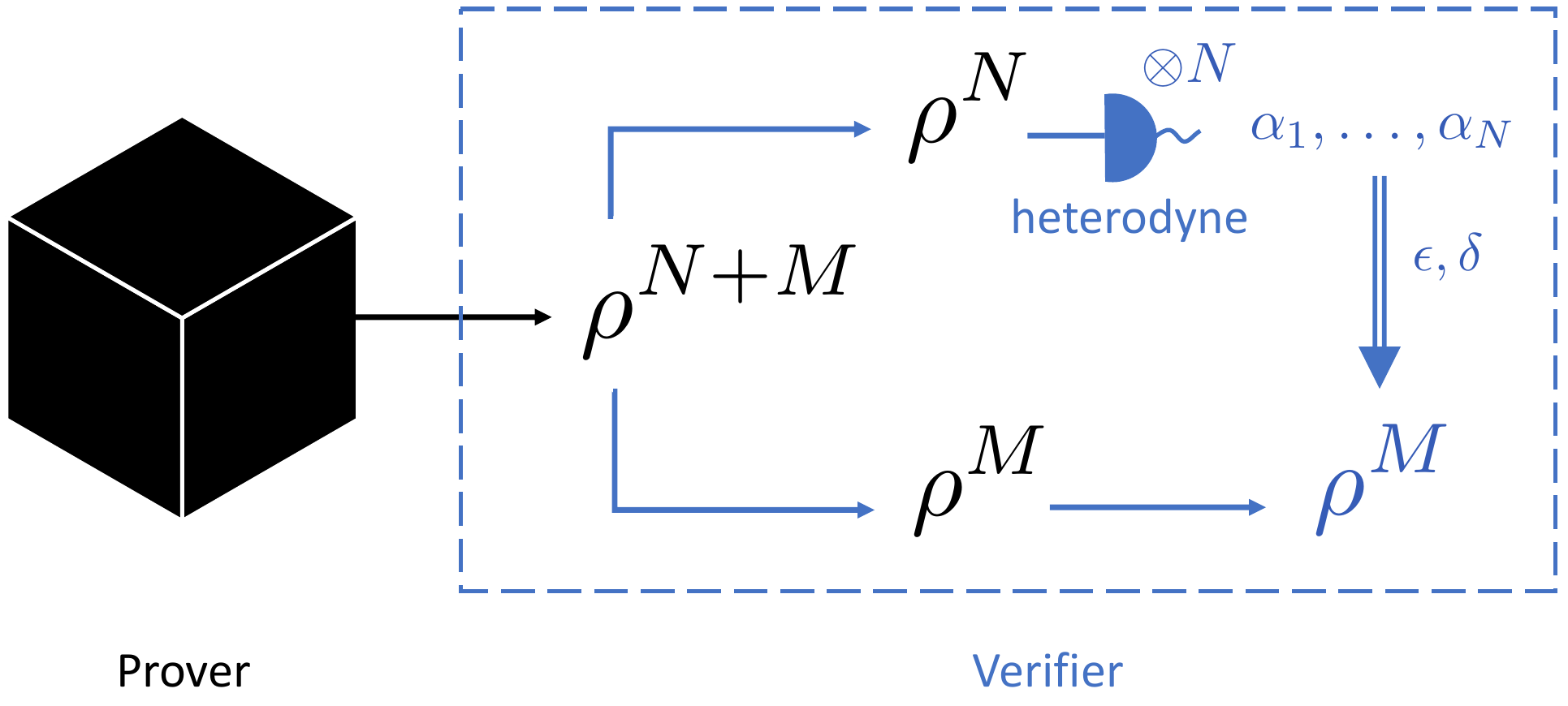}
		\caption{A pictorial representation of the structure of our protocols. The verifier (within the blue dashed rectangle) asks the prover (black box) for $N+M$ copies of some target continuous-variable pure quantum state. The prover sends a (mixed) quantum state $\rho^{N+M}$ over $N+M$ subsystems (under the i.i.d.\ assumption, this state is of the form $\rho^{\otimes N+M}$). The verifier measures $N$ subsystems of $\rho^{N+M}$ at random with heterodyne detection and analyses the obtained samples to check if the remaining state $\rho^M$ over $M$ subsystems is close to $M$ copies of the target state, with a precision parameter $\epsilon>0$ and a confidence parameter $\delta>0$. If the answer is positive, the verifier may then decide to use the verified state $\rho^M$ for any quantum information processing task. In an experimental scenario, the blackbox prover corresponds to the experiment while the experimenter plays the role of the verifier; in that case, the state $\rho^{N+M}$ is the output quantum state over $N+M$ runs of the experiment.}
		\label{fig:setup}
	\end{center}
\end{figure}

\textit{Firstly}, generalising several results from~\cite{paris1996quantum,chabaud2019building}, we derive a protocol for reliably estimating the fidelity of an unknown continuous-variable quantum state with any single-mode continuous-variable quantum state of bounded support over the Fock basis using heterodyne detection (Protocol~\ref{prot:fe}).

\textit{Secondly}, using this single-mode fidelity estimation protocol as a subroutine, we obtain a protocol for reliably estimating a fidelity witness, i.e., a tight lower bound on the fidelity, for a large class of multimode continuous-variable quantum states using heterodyne detection of an unknown state (Protocol~\ref{prot:we}). For $m$ modes, this class of certifiable states corresponds to the $m$-mode states of the form:
\be
\left(\bigotimes_{i=1}^m\hat G_i\right)\hat U\left(\bigotimes_{i=1}^m\ket{C_i}\right)\!,
\label{class}
\ee
where $\hat U$ is an $m$-mode passive linear transformation (a unitary transformation of the creation and annihilation operators of the modes), where each state $\ket{C_i}$ is a single-mode pure state with bounded support over the Fock basis, and where each operation $\hat G_i$ is a single-mode Gaussian unitary, for all $i\in\{1,\dots,m\}$. 

\textit{Thirdly}, using this multimode fidelity witness estimation protocol as a subroutine, we obtain a protocol for verifying the output states of Boson Sampling experiments (Protocol~\ref{prot:bs}). Our protocol provides a certificate on the total variation distance between the experimental and target probability distributions for any observable, efficiently in the number of modes and input photons, thus enabling a convincing demonstration of quantum speedup with photonic quantum computing. Moreover, the verification protocol is implemented within the same experimental setup, simply by replacing the output detectors by balanced heterodyne detection. We also optimise the efficiency of the protocol to the specific setting of Boson Sampling.

Note that other fidelity witnesses for multimode photonic state preparations have been introduced in~\cite{aolita2015reliable}. However, the number of measurements needed to estimate with constant precision the output state of a Boson Sampling interferometer with $n$ input photons over $m$ modes with their witnesses scales as $\Omega(m^{n+4})$, under the i.i.d.\ assumption, while we show that our protocol provides the same precision with $O(m^2\log m)$ measurements. Moreover, we are able to remove the i.i.d.\@ state preparation assumption, at the cost of an increased---though still polynomial---number of measurements needed for the same estimate precision and confidence interval.

The rest of the paper is structured as follows: we recall heterodyne detection in section~\ref{sec:het}, we present our single-mode fidelity estimation protocol in section~\ref{sec:feprotocol}, we introduce our multimode fidelity witness estimation protocol in section~\ref{sec:weprotocol} and we discuss the verification of Boson Sampling in section~\ref{sec:BS}. Finally, we conclude in section~\ref{sec:conclusion}.

%--------------------------------------------------------------------------------

\section{Heterodyne detection}
\label{sec:het}

\noindent We write $\mathbb N^*$ the set of positive integers. Hereafter, $m\in\mathbb N^*$ denotes the number of modes and $\{\ket n\}_{n\in\mathbb N}$ is the single-mode Fock basis. Let us define the single-mode displacement operator $\hat D(\beta)=\exp[\beta\hat a^\dag-\beta^*\hat a]$, for all $\beta\in\mathbb C$, and the single-mode squeezing operator $\hat S(\xi)=\exp[\frac12(\xi^*\hat a^2-\xi\hat a^{\dag2})]$, for all $\xi\in\mathbb C$, where $\hat a^\dag$ and $\hat a$ are the creation and annihilation operators of the mode~\cite{Lloyd2012}, respectively. We use bold math for multimode notations. For all $\bm\beta,\bm\xi\in\mathbb C^m$, we write $\hat D(\bm\beta)=\bigotimes_{i=1}^m\hat D(\beta_i)$ and $\hat S(\bm\xi)=\bigotimes_{i=1}^m\hat S(\xi_i)$.

Heterodyne detection, also called double homodyne, dual-homodyne or eight-port homodyne detection~\cite{ferraro2005gaussian}, is a single-mode Gaussian measurement, projecting onto (unnormalised) coherent states. Mathematically, measuring a state $\rho$ with heterodyne detection amounts to sampling from its Husimi $Q$ phase space quasiprobability function~\cite{yuen1980optical}:
\be
Q_\rho(\alpha)=\frac1\pi\braket{\alpha|\rho|\alpha},
\ee
where $\ket\alpha=e^{-\frac{|\alpha|^2}2}\sum_{n\ge0}{\frac{\alpha^n}{\sqrt{n!}}\ket n}$ is the coherent state of amplitude $\alpha\in\mathbb C$. Heterodyne detection corresponds to a simultaneous noisy measurement of conjugate quadratures and can be implemented experimentally by mixing the state to be measured with the vacuum on a balanced beam splitter and detecting both output branches with homodyne detection. 

\begin{figure}
	\begin{center}
		\includegraphics[width=0.7\columnwidth]{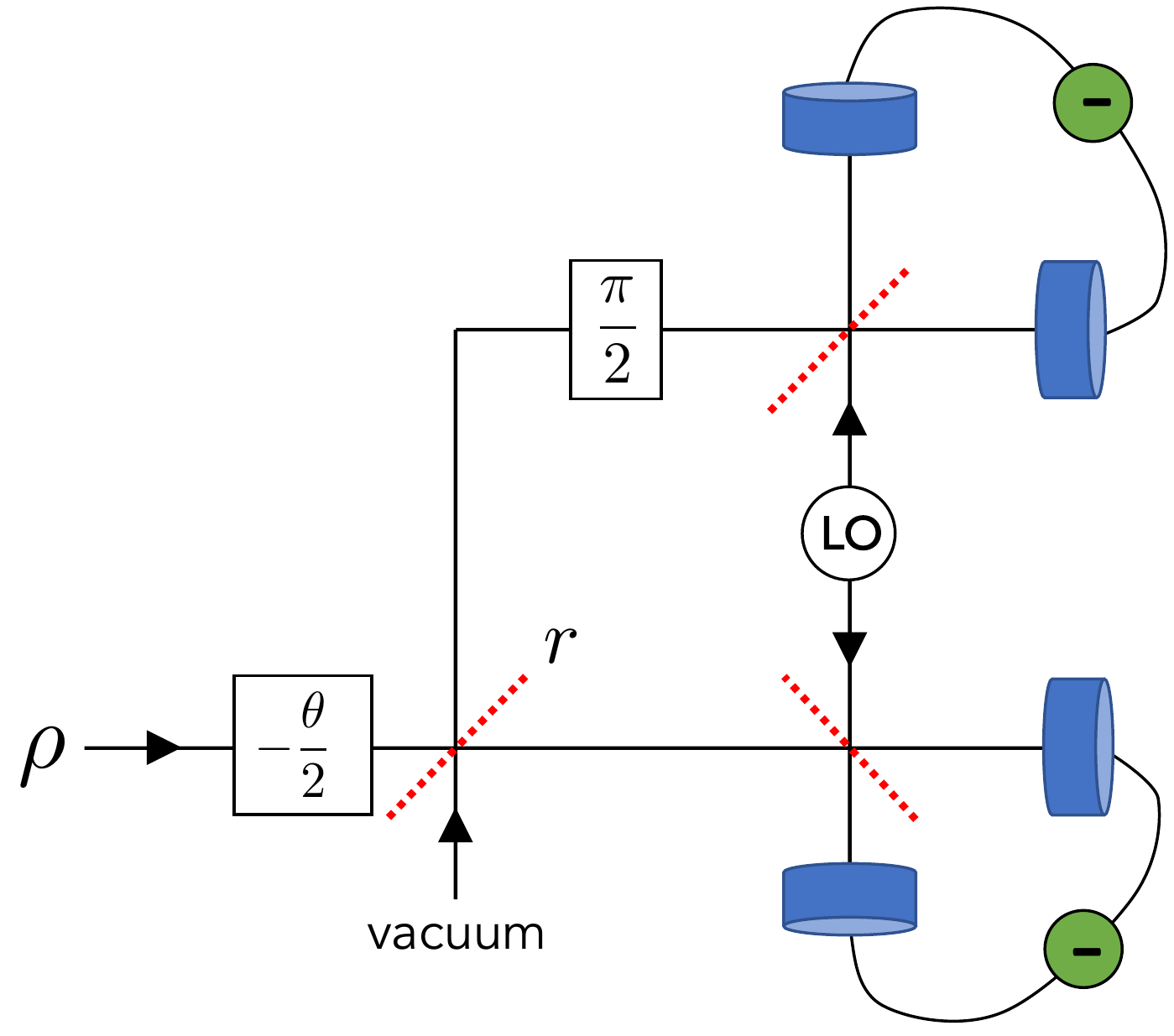}
		\caption{A schematic representation of optical single-mode unblanced heterodyne detection of a state $\rho$, with unbalancing parameter $\xi=re^{i\theta}$. The dashed red lines represent beam splitters. LO stands for local oscillator, i.e., a strong coherent state. The blue cylinders are photodiode detectors.}
		\label{fig:2homodyne}
	\end{center}
\end{figure}

Unbalancing the beam splitter and changing the phase of the local oscillator in the homodyne detection also allows the measurement of squeezed quadratures rotated in phase space (Fig.~\ref{fig:2homodyne}). Note that the phase of the local oscillator can be modulated with a classical post-processing step. We refer to this Gaussian measurement as unbalanced heterodyne detection with unbalancing parameter $\xi\in\mathbb C$.
The positive operator-valued measure (POVM) elements of a tensor product of single-mode unbalanced heterodyne detection over $m$ modes with unbalancing parameters $\bm\xi=(\xi_1,\dots,\xi_m)\in\mathbb C^m$ are given by
\be
\Pi_{\bm\alpha}^{\bm\xi}=\frac1{\pi^m}\ket{\bm\alpha,\bm\xi}\!\bra{\bm\alpha,\bm\xi},
\ee
for all $\bm\alpha=(\alpha_1,\dots,\alpha_m)\in\mathbb C^m$, where $\ket{\bm\alpha,\bm\xi}=\bigotimes_{i=1}^m{\ket{\alpha_i,\xi_i}}$ is a tensor product of squeezed coherent states $\hat S(\xi_i)\hat D(\alpha_i)\ket0$. Writing $\xi=re^{i\theta}$, the unbalalancing parameter is related to the optical setup by $r=\log\left(\frac TR\right)$, where $T$ and $R$ are the reflectance and transmittance of the unbalanced beam splitter, respectively, with the phase of the local oscillator being $-\frac\theta2$ (see Appendix~\ref{app:hetPOVM}).

With these properties of heterodyne detection laid out, we detail in the following section our single-mode fidelity estimation protocol based on heterodyne detection.

%--------------------------------------------------------------------------------

\section{Single-mode fidelity estimation}
\label{sec:feprotocol}

\noindent Normalised single-mode pure quantum states with bounded support over the Fock basis are referred to as \textit{core states}~\cite{menzies2009gaussian,chabaud2020stellar}. In what follows, we introduce a protocol for estimating the fidelity of any single-mode continuous-variable (mixed) quantum state with any target core state using heterodyne detection, by generalising heterodyne fidelity estimates from~\cite{paris1996quantum,chabaud2019building}. We give two versions of our protocol, with and without the i.i.d.\ assumption.

We denote by $\underset{\alpha\leftarrow D}{\mathbb E}[f(\alpha)]$ the expected value of a function $f$ for samples drawn from a distribution $D$.
Let us introduce for $k,l\ge0$ the polynomials
\be
\ba
\mathcal{L}_{k,l}(z)&=e^{zz^*}\frac{(-1)^{k+l}}{\sqrt{k!}\sqrt{l!}}\frac{\partial^{k+l}}{\partial z^k\partial z^{*l}}e^{-zz^*}\\
&=\sum_{p=0}^{\min{(k,l)}}{\frac{\sqrt{k!}\sqrt{l!}(-1)^p}{p!(k-p)!(l-p)!}z^{l-p}z^{*k-p}},
\ea
\label{2DL}
\ee
for $z\in\mathbb C$, which are, up to a normalisation, the Laguerre $2$D polynomials appearing in the expression of quasiprobability distribution of Fock states~\cite{wunsche1998laguerre}. Following~\cite{chabaud2019building}, for all $k,l\in\mathbb N$, we introduce with these polynomials the functions
\be
\ba
f_{k,l}(z,\eta)=\frac1{\eta^{1+\frac{k+l}2}} e^{\left(1-\frac{1}{\eta}\right)zz^*}\mathcal{L}_{l,k}\left(\frac z{\sqrt{\eta}}\right),
\label{fkl}
\ea
\ee
for all $z\in\mathbb C$ and all $0<\eta<1$. Now let us define, for all $p\in\mathbb N^*$, all $z\in\mathbb C$ and all $0<\eta<1$, the generalised functions:
\be
\ba
g_{k,l}^{(p)}(z,\eta):=\sum_{j=0}^{p-1}&(-1)^j\eta^jf_{k+j,l+j}(z,\eta)\\
&\times\sqrt{\binom{k+j}k\binom{l+j}l}.
\ea
\label{g}
\ee
The functions $z\mapsto g_{k,l}^{(p)}(z,\eta)$ are polynomials multiplied by a converging Gaussian function and are thus bounded over $\mathbb C$. 
Using these functions, we derive the following result:

\begin{lem} \label{lem:main}
Let $p\in\mathbb N^*$, let $k,l\in\mathbb N$ and let $0<\eta\le\frac{p+1}{p\sqrt{(k+p+1)(l+p+1)}}$. Let $\rho=\sum_{i,j=0}^{+\infty}{\rho_{ij}\ket i\!\bra j}$ be a single-mode density operator.
Then,
\be
\left|\rho_{kl}-\!\!\underset{\alpha\leftarrow Q_{\mathrlap\rho}}{\mathbb E}[g_{k,l}^{(p)}(\alpha,\eta)]\right|\le\eta^p\sqrt{\binom{k+p}k\binom{l+p}l},
\ee
where the function $g_{k,l}^{(p)}$ is defined in Eq.~(\ref{g}).
\end{lem}

\noindent We give a proof in Appendix~\ref{app:lemmain}. In particular, since heterodyne detection corresponds to sampling from the Husimi $Q$ function of the measured state, this result enables estimating an element $k,l$ of the density matrix from samples of heterodyne detection by computing the mean of the function $z\mapsto g_{k,l}^{(p)}(z,\eta)$ over these samples. These heterodyne estimates generalise both the ones introduced in~\cite{paris1996quantum}, which we retrieve in the limit $\eta\to1$ and for which an assumption on the size of the support of the measured state is necessary, and the ones introduced in \cite{chabaud2019building}, which we retrieve by setting $p=1$ and which yield less efficient estimates.

Additionnally, for any core state $\ket C=\sum_{n=0}^{c-1}{c_n\ket n}$ where $c\in\mathbb N^*$,  and for all $p\in\mathbb N^*$, all $0<\eta<1$ and all $z\in\mathbb C$, we define:
\be
g_C^{(p)}(z,\eta):=\sum_{0\le k,l\le c-1}{c_k^*c_l\,g_{k,l}^{(p)}(z,\eta)}.
\label{gCmain}
\ee
Computing the mean of the function $z\mapsto g_C^{(p)}(z,\eta)$ over samples from heterodyne detection of a state allows us to estimate the fidelity of the measured state with the core state $\ket C$.  The estimation error may be controlled by optimising over the choice of the free parameters $p$ and $\eta$.  In particular, picking a smaller value for $\eta$ decreases the error in Lemma~\ref{lem:main} and increased the range of the function $z\mapsto g_C^{(p)}(z,\eta)$, thus increasing the statistical error.

Based on this result, we now present the version of our single-mode fidelity estimation protocol under i.i.d.\@ assumption and discuss afterwards its version without i.i.d.\@ assumption:

\begin{prot}[Single-mode fidelity estimation]\label{prot:fe}
Let $c\in\mathbb N^*$ and let $\ket C=\sum_{n=0}^{c-1}{c_n\ket n}$ be a core state. Let also $N,M\in\mathbb N^*$, and let $p\in\mathbb N^*$ and $0<\eta<1$ be free parameters. Let $\rho^{\otimes N+M}$ be $N+M$ copies of an unknown single-mode (mixed) quantum state $\rho$.
\begin{enumerate}
\item
Measure $N$ copies of $\rho$ with heterodyne detection, obtaining the samples $\alpha_1,\dots,\alpha_N\in\mathbb C$.
\item
Compute the mean $F_C(\rho)$ of the function $z\mapsto g_C^{(p)}(z,\eta)$ (defined in Eq.~(\ref{gCmain})) over the samples $\alpha_1,\dots,\alpha_N\in\mathbb C$.
\item
Compute the fidelity estimate $F_C(\rho)^M$.
\end{enumerate}
\end{prot}

\noindent The value $F_C(\rho)^M$ obtained is an estimate of the fidelity between the $M$ remaining copies of $\rho$ and $M$ copies of the target core state $\ket C$. The efficiency of the protocol is summarised by the following result:

\begin{theo}\label{th:fe}
Let $\epsilon,\delta>0$. With the notations of Protocol~\ref{prot:fe}, the estimate $F_C(\rho)^M$ is $\epsilon$-close to the fidelity $F(\ket C^{\otimes M},\rho^{\otimes M})$ with probability $1-\delta$ whenever $N\ge N_1$, with 
\be
N_1=\mathcal O\!\left(\left(\frac M\epsilon\right)^{2+\frac{2c}p}\log\left(\frac1\delta\right)\right)\!,
\ee
where the constant prefactor depends on the choice of the free parameters $p$ and $\eta$ in Protocol~\ref{prot:fe}.
\end{theo}

\noindent We give a detailed version of the theorem together with its proof in Appendix~\ref{app:thfe}, which combines Lemma~\ref{lem:main} with Hoeffding's inequality~\cite{hoeffding1963probability}. Since the parameter $p\in\mathbb N^*$ may vary freely, the scaling can be brought arbitrarily close to $\mathcal O((\frac M\epsilon)^2\log(\frac1\delta))$. In practice, since the constant prefactor increases with $p$, an optimisation yields the optimal choice for $p$ which minimises $N$ when $M$, $\epsilon$ and $\delta$ are fixed. Moreover, the choice of the other free parameter $0<\eta<1$ contributes to minimising the constant prefactor and the efficiency of Protocol~\ref{prot:fe} can also be refined by taking into account the expression of the single-mode target core state in the Fock basis (see~\cite{chabaud2021certification} for a detailed analysis).

Hence, Protocol~\ref{prot:fe} gives a reliable and efficient way of estimating the fidelity of an unknown single-mode continuous-variable quantum state with any target core state. These core states play an important role in the characterisation of single-mode non-Gaussian states~\cite{chabaud2020stellar}. Additionnally, they form a dense subset (for the trace norm) of the set of normalised single-mode pure quantum states. Hence, given any target normalised single-mode pure state $\ket\psi$, we can use our fidelity estimation protocol by targeting instead a truncation of the state $\ket\psi$ in the Fock basis in order to estimate the fidelity of any  single-mode continuous-variable quantum state with the state $\ket\psi$ using heterodyne detection.

In an experimental scenario, one may set $M=1$ in Protocol~\ref{prot:fe} and deduce information about the output state of a single run of the experiment. In a cryptographic scenario, where the state $\rho$ is provided by an untrusted prover, the verifier may want to use some copies of the state to perform fidelity estimation of the remaining $M$ copies, before using these remaining states for subsequent quantum information processing. In that context, the i.i.d\ assumption may no longer be justified: the quantum state sent by the prover $\rho^{N+M}$ over $N+M$ subsystems may no longer be of the form $\rho^{\otimes N+M}$, as all subsystems can possibly be entangled with each other (as well as with other quantum systems on the side of the prover). Hence, we extend Protocol~\ref{prot:fe} to a version which does not make the i.i.d.\ assumption, using a de Finetti reduction from~\cite{renner2009finetti,chabaud2019building}. The non-i.i.d.\ version of Protocol~\ref{prot:fe} obtained is nearly identical, up to slight differences in the classical post-processing: a small fraction of the measured subsystems have to be discarded at random and another fraction of the samples are used for an energy test. This comes at the cost of an increased number of measurements which corresponds however to a polynomial overhead for a polynomial precision and confidence. We give a detailed analysis in Appendix~\ref{app:fenotiid}.

%--------------------------------------------------------------------------------

\section{Multimode fidelity witness estimation}
\label{sec:weprotocol}

\noindent In this section, we extend our single-mode fidelity estimation protocol from the previous section to the multimode case, and we show that being able to estimate single-mode fidelities with heterodyne detection is sufficient to estimate tight fidelity witnesses, i.e., tight lower bounds on the fidelity, for the large class of multimode states in Eq.~(\ref{class}).

Note that the transformation in Eq.~(\ref{class}) can always be expressed as $\hat S(\bm\xi)\hat D(\bm\beta)\,\hat U$, where $\bm\xi,\bm\beta\in\mathbb C^m$, since single-mode Gaussian unitaries can be decomposed as a displacement, a squeezing with complex parameter, and a phase-space rotation~\cite{Lloyd2012}. Note also that an $m$-mode passive linear transformation $\hat U$ is associated with an $m\times m$ unitary matrix $U$ which describes its action on the creation and annihilation operators of the modes~\cite{ferraro2005gaussian}. 

Our extension to the multimode case is based on two observations (Lemmas~\ref{lem:product} and~\ref{lem:hetmagic} hereafter).

\textit{Firstly}, if all the single-mode subsystems $\rho_i$ of a multimode quantum state $\bm\rho$ are close enough to single-mode pure states, then $\bm\rho$ is close to the tensor product of these single-mode pure states. Formally:
 
\begin{lem}\label{lem:product}
Let $\bm\rho$ be a state over $m$ subsystems. For all $i\in\{1,\dots,m\}$, we denote by $\rho_i=\Tr_{\{1,\dots,m\}\setminus\{i\}}(\bm\rho)$ the reduced state of $\bm\rho$ over the $i^{th}$ subsystem. Let $\ket{\psi_1},\dots,\ket{\psi_m}$ be pure states. For all $i\in\{1,\dots,m\}$, we write
\be
W:=1-\sum_{i=1}^m{\left(1-F(\rho_i,\psi_i)\right)},
\ee 
where $F$ is the fidelity, and $\ket{\bm\psi}=\ket{\psi_1}\otimes\dots\otimes\ket{\psi_m}$. Then,
\be
1-m(1-F(\bm\rho,\bm\psi))\le W\le F(\bm\rho,\bm\psi).
\label{lemproductmain}
\ee
%&
\end{lem}

\noindent We give a proof in Appendix~\ref{app:lemproduct}.  This result shows that $W$ is a tight lower bound on the fidelity. Importantly, this fidelity witness is only a function of the single-mode fidelities. In particular, being able to estimate single-mode fidelities with single-mode pure states using balanced heterodyne detection is enough to estimate fidelity witnesses with pure product states using balanced heterodyne detection. 

At this point, we can estimate tight fidelity witnesses for pure product states using Protocol~\ref{prot:fe} for each of the single-mode subsystems in parallel. We make use of the properties of heterodyne detection in order to further extend the class of target states for which a tight fidelity witness can be efficiently obtained. 

\textit{Secondly}, a passive linear transformation followed by single-mode Gaussian unitary operations before unbalanced heterodyne detection is equivalent to performing balanced heterodyne detection directly, then post-processing efficiently the classical samples. Formally:

\begin{lem}\label{lem:hetmagic} Let $\bm\beta,\bm\xi\in\mathbb C^m$ and let $\hat V=\hat S(\bm\xi)\hat D(\bm\beta)\,\hat U$, where $\hat U$ is an $m$-mode passive linear transformation with $m\times m$ unitary matrix $U$. For all $\bm\gamma\in\mathbb C^m$, let $\bm\alpha=U^\dag(\bm\gamma-\bm\beta)$. Then,
\be
\Pi^{\bm\xi}_{\bm\gamma}=\hat V\Pi^{\bm0}_{\bm\alpha}\hat V^\dag.
\ee
\end{lem}

\noindent We give a proof in Appendix~\ref{app:lemhetmagic}. A striking consequence is that a certain class of quantum operations before a balanced heterodyne measurement can be inverted efficiently by unbalancing the detection and perfoming classical operations on the samples instead. This result generalises a technique used in continuous-variable quantum key distribution protocols based on heterodyne detection, as well as in a scheme for homomorphic encryption of linear optics quantum computation~\cite{ouyang2020homomorphic}, which allows one to apply a random unitary operation at the level of the classical samples (i.e., a multiplication of the vector of samples by a unitary matrix) rather than directly on the quantum state~\cite{leverrier2013security}.
In particular, for such a transformation $\hat V$, if a multimode pure product state $\bigotimes_{i=1}^m{\ket{\psi_i}}$ can be efficiently verified using balanced heterodyne detection, then the state $\hat V\bigotimes_{i=1}^m{\ket{\psi_i}}$ can be efficiently verified using unbalanced heterodyne detection by post-processing classically the samples obtained.

This leads us to our multimode fidelity witness estimation protocol. Like for Protocol~\ref{prot:fe}, we give the version of our protocol under i.i.d.\@ assumption, and discuss afterwards its version without i.i.d.\@ assumption:

\begin{prot}[Multimode fidelity witness estimation]\label{prot:we}
Let $c_1,\dots,c_m\in\mathbb N^*$. Let $\ket{C_i}=\sum_{n=0}^{c_i-1}{c_{i,n}\ket n}$ be a core state, for all $i\in\{1,\dots,m\}$. Let $\hat U$ be an $m$-mode passive linear transformation with $m\times m$ unitary matrix $U$, and let $\bm\beta,\bm\xi\in\mathbb C^m$. We write $\ket{\bm\psi}=\hat S(\bm\xi)\hat D(\bm\beta)\,\hat U\bigotimes_{i=1}^m\ket{C_i}$ the $m$-mode target pure state. Let $N,M\in\mathbb N^*$, and let $p_1,\dots,p_m\in\mathbb N^*$ and $0<\eta_1,\dots,\eta_m<1$ be free parameters. Let $\bm\rho^{\otimes N+M}$ be $N+M$ copies of an unknown $m$-mode (mixed) quantum state $\bm\rho$.
\begin{enumerate}
\item
Measure all $m$ subsystems of $N$ copies of $\bm\rho$ with unbalanced heterodyne detection with unbalancing parameters $\bm\xi=(\xi_1,\dots,\xi_m)$, obtaining the vectors of samples $\bm\gamma^{(1)},\dots,\bm\gamma^{(N)}\in\mathbb C^m$.
\item
For all $k\in\{1,\dots,N\}$, compute the vectors $\bm\alpha^{(k)}=U^\dag\left(\bm\gamma^{(k)}-\bm\beta\right)$. We write $\bm\alpha^{(k)}=(\alpha^{(k)}_1,\dots\alpha^{(k)}_m)\in\mathbb C^m$.
\item
For all $i\in\{1,\dots,m\}$, compute the mean $F_{C_i}(\bm\rho)$ of the function $z\mapsto g_{C_i}^{(p_i)}(z,\eta_i)$ (defined in Eq.~(\ref{gCmain})) over the samples $\alpha_i^{(1)},\dots,\alpha_i^{(N)}\in\mathbb C$.
\item
Compute the fidelity witness estimate $W_{\bm\psi}=1-\sum_{i=1}^m{(1-F_{C_i}(\bm\rho)^M)}$.
\end{enumerate}
\end{prot}

\noindent The value $W_{\bm\psi}$ obtained is an estimate of a tight lower bound on the fidelity between the $M$ remaining copies of $\bm\rho$ and $M$ copies of the $m$-mode target state $\ket{\bm\psi}=\hat S(\bm\xi)\hat D(\bm\beta)\,\hat U\bigotimes_{i=1}^m\ket{C_i}$. The efficiency and completeness of the protocol are summarised by the following result:

\begin{theo}\label{th:we}
Let $\epsilon,\delta>0$. With the notations of Protocol~\ref{prot:we}, for all $i\in\{1,\dots,m\}$, let $n_i$ be the number of indices $j\in\{1,\dots,m\}$ such that $\frac{c_j}{p_j}=\frac{c_i}{p_i}$. Then,
\be
1\!-m(1\!-\!F(\ket{\bm\psi}^{\!\otimes M}\!\!\!\!\!,\bm\rho^{\otimes M})-\epsilon\!\le\! W_{\bm\psi}\!\le\!F(\ket{\bm\psi}^{\!\otimes M}\!\!\!\!\!,\bm\rho^{\otimes M})+\epsilon,
\label{tightW}
\ee
with probability $1-\delta$, whenever $N\ge N_2$, with 
\be
N_2=\mathcal O\!\left(\max_i\left\{\left(\frac{Mm}\epsilon\right)^{2+\frac{2c_i}{p_i}}\log(n_i)\log\left(\frac1\delta\right)\right\}\right)\!,
\ee
where the constant prefactor depends on the choice of the free parameters $p_i$ and $\eta_i$ in Protocol~\ref{prot:we}. In particular, if the tested state is perfect, i.e., $\rho^{\otimes N+M}=\ket\psi\!\bra\psi^{\otimes N+M}$, then $W_{\bm\psi}\ge1-\epsilon$ with probability $1-\delta$.
\end{theo}

\noindent We give a detailed version of the theorem together with its proof in Appendix~\ref{app:thwe}, which combines Theorem~\ref{th:fe} with Lemmas~\ref{lem:product} and~\ref{lem:hetmagic}.

Note that Protocol~\ref{prot:we} uses Protocol~\ref{prot:fe} as a subroutine.
Writing $p=\min_ip_i$ and $c=\max_ic_i$, with $n_i\le m$ we obtain $N_2=\mathcal O((\frac{Mm}\epsilon)^{2+\frac{2c}p}\log(m)\log(\frac1\delta))$. In particular, compared to the single-mode case, the overhead for certifying $m$-mode states is only $\mathcal O(m^{2+\frac{2c}p}\log m)$, which is a huge improvement compared to, e.g., the exponential scaling of tomography~\cite{d2003quantum}. Since the parameter $p\in\mathbb N^*$ may vary freely, the scaling can be brought arbitrarily close to $\mathcal O((\frac{Mm}\epsilon)^2\log(m)\log(\frac1\delta))$. The constant prefactor increases with $p$, and an optimisation can yield the optimal choice for $p$ which minimises $N_2$ when $M$, $m$, $\epsilon$ and $\delta$ are fixed. Moreover, the choice of the other free parameters $0<\eta_i<1$ also contributes to minimising the constant prefactor.  Like for Protocol~\ref{prot:fe}, the efficiency of Protocol~\ref{prot:we} can also be refined by taking into account the expression of the single-mode target core states in the Fock basis.

Without the i.i.d.\ assumption, we obtain an equivalent protocol by using the version of our single-mode fidelity estimation protocol which does not assume i.i.d.\ state preparation for estimating the single-mode fidelities. Like for Protocol~\ref{prot:fe}, the final protocol is nearly identical, up to slight differences in the classical post-processing. Removing the i.i.d.\ assumption then comes at the cost of an increased number of measurements necessary for a polynomial witness precision and confidence interval, corresponding to a polynomial overhead. We give a detailed analysis in Appendix~\ref{app:wenotiid}.

The class of states for which fidelity witnesses can be efficiently estimated using Protocol~\ref{prot:we} includes classically intractable quantum states, such as the output states of Boson Sampling experiments, as we detail in the next section.

%--------------------------------------------------------------------------------

\section{Verification of Boson Sampling}
\label{sec:BS}

\noindent The setup of Boson Sampling with input single-photons consists in $n$ single-photon states fed into an interferometer over $m$ modes, together with the vacuum state over $m-n$ modes, and measured either with single-photon threshold detection~\cite{Aaronson2013} or with unbalanced heterodyne detection~\cite{chabaud2017continuous,chakhmakhchyan2017boson} (Fig.~\ref{fig:CVSverif}). The output probabilities are related to permanents of complex matrices, which are hard to compute~\cite{valiant1979complexity} and even to approximate~\cite{Aaronson2013}. A striking consequence is that classical computers cannot sample exactly from the corresponding probability distributions unless the polynomial hierarchy collapses, contradicting a widely believed conjecture from complexity theory. Aaronson and Arkhipov importantly showed that even an approximate version of their model is still hard to sample for classical computers, provided two additional conjectures on the permanent of random Gaussian matrices hold true~\cite{Aaronson2013}. More precisely, they showed under these conjectures that sampling from a probability distribution that has small constant total variation distance with an ideal Boson Sampling distribution is classically hard in the so-called antibunching regime $n=\mathcal O(\sqrt m)$ (in fact they used $m=\Omega(n^6)$ and conjectured that $m=\Omega(n^2)$ was sufficient, as it was later shown~\cite{jiang2019distances}).

The output state of a Boson Sampling experiment with $n$ photons fed into an interferometer over $m$ modes, described by an $m$-mode passive linear transformation $\hat U$ reads $\hat U\ket{1\dots10\dots0}$. This state is of the form of Eq.~(\ref{class}), with $\hat G_i=\mathbb 1$ for all $i\in\{1,\dots,m\}$, $\ket{C_i}=\ket1$ for $i\in\{1,\dots,n\}$ and $\ket{C_j}=\ket0$ for $j\in\{n+1,\dots,m\}$.
Our fidelity witness estimation protocol from the previous section can thus be applied to certify efficiently the fidelity of the output state of a Boson Sampling experiment with the ideal target output state, using only balanced heterodyne detection. Then, our protocol yields a certificate of the total variation distance with the target probability distribution for any observable by the Fuchs--Van de Graaf inequality~\cite{fuchs1999cryptographic} and standard properties of the trace distance~\cite{NielsenChuang}. Indeed, for any observable $\hat O$ and any states $\bm\rho$ and $\bm\sigma$,
\be
\|P^{\hat O}_{\bm\rho}-P^{\hat O}_{\bm\sigma}\|\le D(\bm\rho,\bm\sigma)\le\sqrt{1-F(\bm\rho,\bm\sigma)},
\label{tvd}
\ee
where the left hand side is the total variation distance between the probability distributions corresponding to a measurement of the observable $\hat O$ for the states $\bm\rho$ and $\bm\sigma$, $D$ is the trace distance and $F$ the fidelity.

We give the version of our Boson Sampling verification protocol under i.i.d.\ assumption and discuss afterwards its version without i.i.d.\ assumption.

\begin{prot}[Boson Sampling verification]\label{prot:bs}
Let $n=\mathcal O(\sqrt m)$. Let $\hat U$ be an $m$-mode passive linear transformation with $m\times m$ unitary matrix $U$. Let $\ket{\bm\psi}=\hat U\ket{1\dots10\dots0}$ be the $m$-mode Boson Sampling target state, with $n$ input photons over $m$ modes. Let $N,M\in\mathbb N^*$, and let $p\in\mathbb N^*$ even, $0<\eta<1$ and $\lambda>\epsilon>0$ be free parameters. Let $\bm\rho^{\otimes N+M}$ be $N+M$ copies of an unknown $m$-mode (mixed) quantum state $\bm\rho$.
\begin{enumerate}
\item
Measure all $m$ subsystems of $N$ copies of $\bm\rho$ with balanced heterodyne detection, obtaining the vectors of samples $\bm\gamma^{(1)},\dots,\bm\gamma^{(N)}\in\mathbb C^m$.
\item
For all $k\in\{1,\dots,N\}$, compute the vectors $\bm\alpha^{(k)}=U^\dag\bm\gamma^{(k)}$. We write $\bm\alpha^{(k)}=(\alpha^{(k)}_1,\dots\alpha^{(k)}_m)\in\mathbb C^m$.
\item
For all $i\in\{1,\dots,n\}$, compute the mean $F_{i,\ket 1}(\bm\rho)$ of the function $z\mapsto g_{\ket 1}^{(p)}(z,\eta)$ (defined in Eq.~(\ref{gCmain})) over the samples $\alpha_i^{(1)},\dots,\alpha_i^{(N)}\in\mathbb C$, and for all $j\in\{n+1,\dots,m\}$, compute the mean $F_{j,\ket 0}(\bm\rho)$ of the function $z\mapsto g_{\ket 0}^{(p)}(z,\eta)$ over the samples $\alpha_j^{(1)},\dots,\alpha_j^{(N)}\in\mathbb C$.
\item
Compute the fidelity witness estimate $W_{\bm\psi}=1-\sum_{i=1}^n{(1-F_{i,\ket 1}(\bm\rho)^M)}-\sum_{j=n+1}^m{(1-F_{j,\ket 0}(\bm\rho)^M)}$.
\item
Abort if $W_{\bm\psi}<1-\lambda+\epsilon$.
Otherwise, accept and measure all $m$ subsystems of the remaining $M$ copies of $\bm\rho$ with unbalanced heterodyne detection (resp.\ single-photon threshold detection), obtaining the $M$ vectors of samples $(\bm s^{(1)},\dots,\bm s^{(M)})\in\mathbb C^m$ (resp.\ $\in\mathbb N^m$).
\end{enumerate}
\end{prot}

\noindent The efficiency and completeness of the protocol are summarised by the following result:

\begin{theo}\label{th:bs}
Let $\delta>0$. We use the notations of Protocol~\ref{prot:bs}. With probability $1-\delta$, either the protocol aborts or the samples $\bm s^{(1)},\dots,\bm s^{(M)}$ are drawn from a probability distribution which has less than $\sqrt\lambda$ total variation distance with the ideal Boson Sampling probability distribution, whenever $N\ge N_3$, with
\be
N_3=\mathcal O\left(\left(\frac{Mm}\epsilon\right)^2\log(m)\log\left(\frac1\delta\right)\right),
\label{Nbs}
\ee
where the constant prefactor depends on the choice of the free parameter $\eta$ in Protocol~\ref{prot:bs}.
Moreover, if the tested state is perfect, i.e., $\rho^{\otimes N+M}=\ket\psi\!\bra\psi^{\otimes N+M}$, then Protocol~\ref{prot:bs} accepts with probability $1-\delta$.
\end{theo}

\noindent We give a proof of the theorem in Appendix~\ref{app:optiBS}, which follows from optimising the efficiency of Protocol~\ref{prot:we} for Boson Sampling output states, together with Eq.~(\ref{tvd}).
The choice of the free parameters $p\in\mathbb N^*$ even and $0<\eta<1$ contributes to minimising the constant prefactor. Additionnally, the choice of $\epsilon$ between $0$ and $\lambda$ can be made to ensure robustness of the protocol with Eq.~(\ref{tightW}): if $\epsilon$ is closer to $0$ then the protocol will reject a good state with less probability. On the other hand, this increases the number of samples needed with Eq.~(\ref{Nbs}). 

Protocol~\ref{prot:bs} allows one to verify a Boson Sampling experiment efficiently, by trusting only single-mode Gaussian measurements. Compared to other experimental platforms, the assumption of trusted measurements is especially relevant in the case of photonic quantum computing, where the measurement apparatus can be easily brought apart from the rest of the experiment, allowing to treat the latter as an untrusted black-box.
Moreover, we obtain a version of Protocol~\ref{prot:bs} without the i.i.d.\ assumption by using the version of Protocol~\ref{prot:we} which does not assume an i.i.d.\ behaviour for the prover from the previous section. This slightly modifies the classical post-processing, requiring the verifier to discard at random a small fraction of the samples and perform an energy test using another fraction of the samples. This comes at the cost of a polynomial overhead for the number of measurements needed for a polynomial confidence, and we refer to Appendix~\ref{app:wenotiid} for the detailed derivation.

\begin{figure}
	\begin{center}
		\includegraphics[width=\columnwidth]{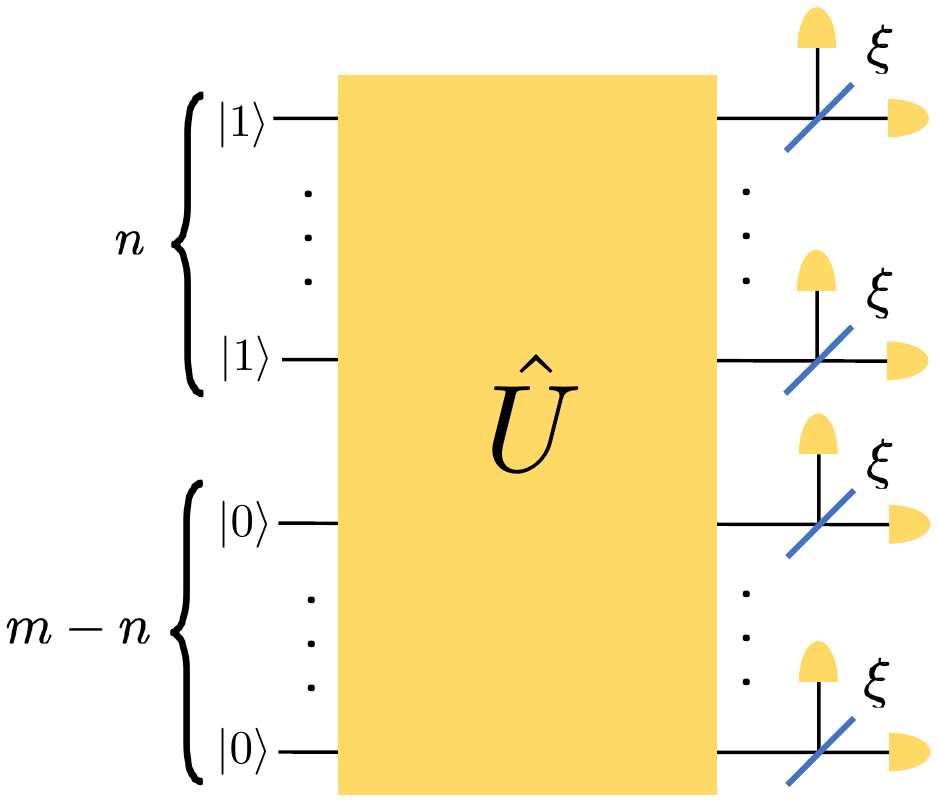}
		\caption{Boson Sampling with interferometer $U$ with $n$ photons over $m$ modes and heterodyne detection with reconfigurable unbalancing parameter $\xi$. The detectors represented are homodyne detectors. Setting $\xi=0$ (balanced detection) allows to certify efficiently the multimode output state with a fidelity witness. Setting $\xi\neq0$, with $|\xi|=\Omega(2^{-\poly m})$~\cite{chabaud2021continuous} allows to perform efficiently a sampling task which is hard for classical computers, unless the polynomial hierarchy collapses\cite{chabaud2017continuous,chakhmakhchyan2017boson}.}
		\label{fig:CVSverif}
	\end{center}
\end{figure}

By simply changing the unbalancing of heterodyne detection of the output modes of a Boson Sampling interferometer, one can switch between verification of Boson Sampling output states and demonstration of quantum speedup with continuous-variable measurements within the same experimental setup. Indeed, Boson Sampling interferometers with unbalanced heterodyne detection are hard to sample classically when the unbalancing of the heterodyne detection is not too small~\cite{chabaud2017continuous,chakhmakhchyan2017boson}, but their output can be efficiently checked simply by switching to balanced heterodyne detection with Protocol~\ref{prot:bs}. This can be done within the same experimental setup using a tunable beam splitter (Fig.~\ref{fig:CVSverif}). However, showing the hardness of approximate sampling for Boson Sampling with Gaussian measurements is an important step before an experimental demonstration of quantum speedup with continuous-variable Gaussian measurements. 
Alternatively, by switching between balanced heterodyne detection and single-photon threshold detectors, for example with a movable mirror~\cite{abrahao2018continuous}, one can switch between verification of Boson Sampling output states and demonstration of quantum speedup with discrete variable measurements, for which approximate sampling hardness is demonstrated, assuming two conjectures on the permanent of random Gaussian matrices and the conjecture that the polynomial hierarchy does not collapse~\cite{Aaronson2013}.

%--------------------------------------------------------------------------------

\section{Discussion}
\label{sec:conclusion}

\noindent In this work, we have introduced various methods for efficiently retrieving information about single-mode and multimode continuous-variable quantum states. 

We have generalised heterodyne estimates both from~\cite{paris1996quantum} and~\cite{chabaud2019building}, obtaining the best of both worlds: a reliable single-mode fidelity estimation protocol with heterodyne detection which does not make an assumption of bounded support for the measured state, while remaining efficient (Protocol~\ref{prot:fe}).

With this fidelity estimation protocol for single subsystems, we have derived a fidelity witness estimation protocol with heterodyne detection for a large class of quantum states over multiple subsystems (Protocol~\ref{prot:we}). This result stems from two observations: on the one hand, if all the subsystems of a quantum state are close to pure states, then this quantum state is close to the tensor product of the pure states; on the other hand, a large class of quantum operations before heterodyne detection can be reversed via classical post-processing. 

Applying our fidelity witness estimation protocol to the case of Boson Sampling, we have shown how to verify efficiently Boson Sampling experiments using heterodyne detection (Protocol~\ref{prot:bs}). Our verification applies to the original Boson Sampling proposal~\cite{Aaronson2013}, where the sampling giving rise to quantum speedup is performed using single-photon threshold detectors, as well as to related proposals~\cite{chabaud2017continuous,chakhmakhchyan2017boson}, where the sampling is performed using unbalanced heterodyne detection instead. In the latter case, the verification and the computational runs only differ by the unbalancing of the verifier's heterodyne detection. Using a tunable beam splitter, this means that Boson Sampling and its verification can be implemented within the same setup. Note however that there is no rigorous approximate sampling hardness result for Boson Sampling with continuous-variable measurements, which is a crucial point before any experimental demonstration. It will likely require extending anti-concentration and average-case hardness conjectures for the permanent or the hafnian to the loop hafnian~\cite{kruse2019detailed,quesada2019franck,chabaud2020classical}. On the other hand, one can already perform a demonstration of quantum speedup with Boson Sampling and discrete variable measurements, by using balanced heterodyne detection for the verification and single-photon threshold detection for the sampling.

For all our protocols, we have derived a version with i.i.d.\ assumption and another without this assumption. The two versions share the same structure, and given that the latter only comes at the cost of a polynomial overhead in the number of measurements and additional classical post-processing, using the i.i.d.\ version of our protocol for verifying a large-scale Boson Sampling experiment would already provide a convincing evidence of quantum speedup with photonic computing.

To that end, it would be interesting to optimise further our protocols for an experimental implementation, by finding the optimal choice for the free parameters and by analysing the case where the verifier has access to imperfect rather than ideal heterodyne detection~\cite{paris1996quantum}. As previously mentioned, it would also be interesting to show the hardness of approximate sampling for Boson Sampling with Gaussian measurements, as it would allow a verifier restricted to single-mode Gaussian measurements to perform both the verification and the hard sampling task.

Demonstration of quantum speedup is likely to be achieved not via a single experiment but rather through numerous experimental realisations competing with always-improving classical simulation algorithms~\cite{neville2017classical,garcia2019simulating,dalzell2020many,clifford2020faster,barak2020spoofing,huang2020classical}. In that context, but also long after the quantum speedup milestone is achieved, efficient methods for reliable certification of quantum devices will be of utmost importance~\cite{eisert2019quantum}. Beyond the quantum speedup milestone, our methods will find various applications for the reliable certification of multimode continuous-variable quantum states in existing and upcoming experiments.

%--------------------------------------------------------------------------------

\subsection*{Acknowledgments}

\noindent We thank Tom Douce, Anthony Leverrier, Saleh Rahimi-Keshari, Gana\"el Roeland, Mattia Walschaers, Nicolas Treps and Valentina Parigi for interesting discussions. We acknowledge support of the European Union's Horizon 2020 Research and Innovation Program under Grant Agreement No. 820445 (QIA), and the ANR through the ANR-17-CE24-0035 VanQuTe.

%--------------------------------------------------------------------------------

\bibliographystyle{linksen}
\bibliography{bibliography}

%--------------------------------------------------------------------------------

\onecolumn

\newpage

\appendix

%--------------------------------------------------------------------------------

\begin{center}
\LARGE \textbf{Appendix}
\end{center}

In this appendix we provide the proofs for the results stated in the main text.

\section{Unbalanced heterodyne detection}
\label{app:hetPOVM}

We relate the POVM for unbalanced heterodyne detection to the experimental detection setup of Fig.~\ref{fig:2homodyne}. Due to the tensor product structure, we only consider the single-mode setting.  We use the definitions $\hat D(\alpha)=\exp[\alpha\hat a^\dag-\alpha^*\hat a]$, for all $\alpha\in\mathbb C$,  $\hat S(\xi)=\exp[\frac12(\xi^*\hat a^2-\xi\hat a^{\dag2})]$, for all $\xi\in\mathbb C$, and $\hat R(\theta)=\exp[i\theta\hat a^\dag\hat a]$ for all $\theta\in[0,2\pi]$, for the displacement operator, the squeezing operator, and the phase operator, respectively.

When the squeezing parameter is real, the derivation can be found in Ref.~\cite{chabaud2017continuous} (note the difference in the definition of the squeezing operator with ours): the POVM elements for heterodyne detection with unbalancing parameter $r\in\mathbb R$ are given by
\begin{equation}
\Pi_{\alpha}^r=\frac1\pi\ket{\alpha,r}\!\bra{\alpha,r},
\end{equation}
for all $\alpha\in\mathbb C$, where $\ket{\alpha,r}=\hat S(r)\hat D(\alpha)\ket0$. In this case, the squeezing parameter is related to the experimental setup by $r=\log\frac TR$. 

The general case $\xi=re^{i\theta}$ with $\theta\neq0$ is obtained by adding a phase operator $\hat R\left(-\frac\theta2\right)$ in the input and changing the local oscillator amplitude from $\alpha$ to $\alpha e^{-\frac{i\theta}2}$ (note that this may be done at the level of the classical samples, i.e., by multiplying the outcomes of the heterodyne detection by a phase $e^{-\frac{i\theta}2}$). Hence, the POVM element corresponding to the detection setup in Fig.~\ref{fig:2homodyne} is given by
\begin{equation}
\begin{aligned}
\Pi_{\alpha}^\xi&=\hat R\left(\frac\theta2\right)\Pi_{e^{-\frac{i\theta}2}\alpha}^r\hat R\left(\frac\theta2\right)^\dag\\
&=\frac1\pi\hat R\left(\frac\theta2\right)\hat S(r)\hat D(e^{-\frac{i\theta}2}\alpha)\ket0\!\bra0\hat D(e^{-\frac{i\theta}2}\alpha)^\dag\hat S(r)^\dag\hat R\left(\frac\theta2\right)^\dag.
\end{aligned}
\end{equation}
Now $\hat R\left(\frac\theta2\right)\hat a^\dag\hat R\left(-\frac\theta2\right)=e^{\frac{i\theta}2}\hat a^\dag$ and $\hat R\left(\frac\theta2\right)\hat a\hat R\left(-\frac\theta2\right)=e^{-\frac{i\theta}2}\hat a$ so
\begin{equation}
\begin{aligned}
\hat R\left(\frac\theta2\right)\hat S(r)\hat D(e^{-\frac{i\theta}2}\alpha)\ket0&=\hat R\left(\frac\theta2\right)\hat S(r)\hat R\left(-\frac\theta2\right)\hat D(\alpha)\hat R\left(\frac\theta2\right)\ket0\\
&=\hat S(re^{i\theta})\hat D(\alpha)\ket0,
\end{aligned}
\end{equation}
where we used $\hat R\left(\frac\theta2\right)\ket0=\ket0$. Hence $\Pi_{\alpha}^\xi=\frac1\pi\ket{\alpha,\xi}\!\bra{\alpha,\xi}$, which completes the proof.

%--------------------------------------------------------------------------------

\section{Proof of Lemma~\ref{lem:main}}
\label{app:lemmain}

\noindent We recall notations from the main text. Let us introduce for $k,l\ge0$ the polynomials
\be
\ba
\mathcal{L}_{k,l}(z)&=e^{zz^*}\frac{(-1)^{k+l}}{\sqrt{k!}\sqrt{l!}}\frac{\partial^{k+l}}{\partial z^k\partial z^{*l}}e^{-zz^*}\\
&=\sum_{p=0}^{\min{(k,l)}}{\frac{\sqrt{k!}\sqrt{l!}(-1)^p}{p!(k-p)!(l-p)!}z^{l-p}z^{*k-p}},
\ea
\label{app2DL}
\ee
for $z\in\mathbb C$, which are, up to a normalisation, the Laguerre $2$D polynomials. For all $k,l\in\mathbb N$, we give with these polynomials the functions~\cite{chabaud2019building,chabaud2019buildingarXiv}
\be
\ba
f_{k,l}(z,\eta)=\frac1{\eta^{1+\frac{k+l}2}} e^{\left(1-\frac{1}{\eta}\right)zz^*}\mathcal{L}_{l,k}\left(\frac z{\sqrt{\eta}}\right),
\label{appfkl}
\ea
\ee
for all $z\in\mathbb C$ and all $0<\eta<1$. The function $z\mapsto f_{k,l}(z,\eta)$, being a polynomial multiplied by a converging Gaussian function, is bounded over $\mathbb C$.
Now let us define, for all $p\in\mathbb N^*$, all $z\in\mathbb C$ and all $0<\eta<1$, the more general bounded functions
\be
g_{k,l}^{(p)}(z,\eta):=\sum_{j=0}^{p-1}{(-1)^j\eta^jf_{k+j,l+j}(z,\eta)\sqrt{\binom{k+j}k\binom{l+j}l}}.
\label{appg}
\ee
We start by showing the following result, which generalises Lemma~3 of \cite{chabaud2019buildingarXiv}:

\begin{lem}\label{lem:inductiong}
Let $p\in\mathbb N^*$, let $k,l\in\mathbb N$ and let $0<\eta<1$. Let $\rho=\sum_{i,j=0}^{+\infty}{\rho_{ij}\ket i\!\bra j}$ be a density operator.
Then,
\be
\underset{\alpha\leftarrow Q_\rho}{\mathbb E}[g_{k,l}^{(p)}(\alpha,\eta)]=\rho_{kl}+(-1)^{p+1}\sum_{q=p}^{+\infty}{\rho_{k+q,l+q}\eta^q\binom{q-1}{p-1}\sqrt{\binom{k+q}k\binom{l+q}l}},
\ee
where the function $g_{k,l}$ is defined in Eq.~(\ref{appg}).
\end{lem}

\begin{proof}

We prove the lemma by induction on $p\in\mathbb N^*$.
By Lemma~3 of \cite{chabaud2019buildingarXiv}, with $E=+\infty$, we have, for all single-mode (mixed) state $\rho=\sum_{i,j\ge0}{\rho_{ij}\ket i\!\bra j}$, for all $k,l\in\mathbb N$
\be
\underset{\alpha\leftarrow Q_\rho}{\mathbb E}[f_{k,l}(\alpha,\eta)]=\rho_{kl}+\sum_{q=1}^{+\infty}{\rho_{k+q,l+q}\eta^q\sqrt{\binom{k+q}k\binom{l+q}l}}.
\label{Lemma3}
\ee
With Eq.~(\ref{appg}), this gives
\be
\underset{\alpha\leftarrow Q_\rho}{\mathbb E}[g_{k,l}^{(1)}(\alpha,\eta)]=\rho_{kl}+\sum_{q=1}^{+\infty}{\rho_{k+q,l+q}\eta^q\sqrt{\binom{k+q}k\binom{l+q}l}},
\ee
hence the lemma is proven for $p=1$ since $\binom{q-1}0=1$. 

Assuming the lemma holds for some $p\in\mathbb N^*$,
\be
\ba
\underset{\alpha\leftarrow Q_\rho}{\mathbb E}[g_{k,l}^{(p)}(\alpha,\eta)]&=\rho_{kl}+(-1)^{p+1}\sum_{q=p}^{+\infty}{\rho_{k+q,l+q}\eta^q\binom{q-1}{p-1}\sqrt{\binom{k+q}k\binom{l+q}l}}\\
&=\rho_{kl}+(-1)^{p+1}\rho_{k+p,l+p}\eta^p\sqrt{\binom{k+p}k\binom{l+p}l}\\
&\quad\quad\quad+(-1)^{p+1}\sum_{q=p+1}^{+\infty}{\rho_{k+q,l+q}\eta^q\binom{q-1}{p-1}\sqrt{\binom{k+q}k\binom{l+q}l}}.
\ea
\ee
Now by Eq.~(\ref{Lemma3}), with $k=k+p$ and $l=l+p$, we have
\be
\underset{\alpha\leftarrow Q_\rho}{\mathbb E}[f_{k+p,l+p}(\alpha,\eta)]=\rho_{k+p,l+p}+\sum_{q=1}^{+\infty}{\rho_{k+q+p,l+q+p}\eta^q\sqrt{\binom{k+q+p}{k+p}\binom{l+q+p}{l+q}}},
\ee
so that 
\be
\ba
\underset{\alpha\leftarrow Q_\rho}{\mathbb E}\left[g_{k,l}^{(p+1)}(\alpha,\eta)\right]&=\underset{\alpha\leftarrow Q_\rho}{\mathbb E}\left[g_{k,l}^{(p)}(\alpha,\eta)+(-1)^p\eta^p\sqrt{\binom{k+p}k\binom{l+p}l}f_{k+p,l+p}(\alpha,\eta)\right]\\
&=\rho_{kl}+(-1)^{p+1}\rho_{k+p,l+p}\eta^p\sqrt{\binom{k+p}k\binom{l+p}l}\\
&\quad\quad\quad+(-1)^{p+1}\sum_{q=p+1}^{+\infty}{\rho_{k+q,l+q}\eta^q\binom{q-1}{p-1}\sqrt{\binom{k+q}k\binom{l+q}l}}\\
&\quad\quad\quad+(-1)^p\eta^p\sqrt{\binom{k+p}k\binom{l+p}l}\rho_{k+p,l+p}\\
&\quad\quad\quad+(-1)^p\eta^p\sqrt{\binom{k+p}k\binom{l+p}l}\sum_{q=1}^{+\infty}{\rho_{k+q+p,l+q+p}\eta^q\sqrt{\binom{k+q+p}{k+p}\binom{l+q+p}{l+p}}}\\
&=\rho_{kl}+(-1)^{p+1}\sum_{q=p+1}^{+\infty}{\rho_{k+q,l+q}\eta^q\binom{q-1}{p-1}\sqrt{\binom{k+q}k\binom{l+q}l}}\\
&\quad\quad\quad+(-1)^p\eta^p\sqrt{\binom{k+p}k\binom{l+p}l}\sum_{q=1}^{+\infty}{\rho_{k+q+p,l+q+p}\eta^q\sqrt{\binom{k+q+p}{k+p}\binom{l+q+p}{l+p}}}\\
&=\rho_{kl}+(-1)^{p+1}\sum_{q=p+1}^{+\infty}{\rho_{k+q,l+q}\eta^q\binom{q-1}{p-1}\sqrt{\binom{k+q}k\binom{l+q}l}}\\
&\quad\quad\quad+(-1)^p\sqrt{\binom{k+p}k\binom{l+p}l}\sum_{q=p+1}^{+\infty}{\rho_{k+q,l+q}\eta^q\sqrt{\binom{k+q}{k+p}\binom{l+q}{l+p}}}\\
&=\rho_{kl}+(-1)^{p+2}\sum_{q=p+1}^{+\infty}\rho_{k+q,l+q}\eta^q\\
&\quad\quad\quad\times\left[\sqrt{\binom{k+p}k\binom{l+p}l\binom{k+q}{k+p}\binom{l+q}{l+p}}-\binom{q-1}{p-1}\sqrt{\binom{k+q}k\binom{l+q}l}\right].
\ea
\label{inductiong}
\ee
We have
\be
\ba
\binom{k+p}k\binom{l+p}l\binom{k+q}{k+p}\binom{l+q}{l+p}&=\frac{(k+p)!(l+p)!(k+q)!(l+q)!}{k!p!l!p!(k+p)!(q-p)!(l+p)!(q-p)!}\\
&=\frac{(k+q)!(l+q)!}{k!l!p!^2(q-p)!^2}\\
&=\binom qp^2\binom{k+q}k\binom{l+q}l,
\ea
\ee
so the last expression in Eq.~(\ref{inductiong}) simplifies as
\be
\ba
\underset{\alpha\leftarrow Q_\rho}{\mathbb E}\left[g_{k,l}^{(p+1)}(\alpha,\eta)\right]&=\rho_{kl}+(-1)^{p+2}\sum_{q=p+1}^{+\infty}\rho_{k+q,l+q}\eta^q\sqrt{\binom{k+q}k\binom{l+q}l}\left[\binom qp-\binom{q-1}{p-1}\right]\\
&=\rho_{kl}+(-1)^{p+2}\sum_{q=p+1}^{+\infty}\rho_{k+q,l+q}\eta^q\binom{q-1}p\sqrt{\binom{k+q}k\binom{l+q}l},
\ea
\ee
which completes the proof of Lemma~\ref{lem:inductiong}.

\end{proof}

\noindent An immediate application of Lemma~\ref{lem:inductiong} gives
\be
\ba
\left|\underset{\alpha\leftarrow Q_\rho}{\mathbb E}[g_{k,l}^{(p)}(\alpha,\eta)]-\rho_{kl}\right|&\le\sum_{q=p}^{+\infty}{|\rho_{k+q,l+q}|\eta^q\binom{q-1}{p-1}\sqrt{\binom{k+q}k\binom{l+q}l}}\\
&\le\max_{q\ge p}\left(\eta^q\binom{q-1}{p-1}\sqrt{\binom{k+q}k\binom{l+q}l}\right)\sum_{q=p}^{+\infty}{|\rho_{k+q,l+q}|}\\
&\le\max_{q\ge p}\left(\eta^q\binom{q-1}{p-1}\sqrt{\binom{k+q}k\binom{l+q}l}\right)\sum_{q=p}^{+\infty}{\sqrt{\rho_{k+q,k+q}\rho_{l+q,l+q}}}\\
&\le\max_{q\ge p}\left(\eta^q\binom{q-1}{p-1}\sqrt{\binom{k+q}k\binom{l+q}l}\right)\sqrt{\sum_{q=p}^{+\infty}\rho_{k+q,k+q}\sum_{q=p}^{+\infty}\rho_{l+q,l+q}}\\
&\le\max_{q\ge p}\left(\eta^q\binom{q-1}{p-1}\sqrt{\binom{k+q}k\binom{l+q}l}\right),
\ea
\label{inductiong1}
\ee
where we used the fact that $\rho$ is hermitian positive semidefinite in the third line, Cauchy-Schwarz inequality in the fourth line and $\Tr(\rho)=1$ in the last line. This last bound is tight: if $q_0\ge p$ is the integer achieving the maximum in the last equation, then the inequality is an equality for $\rho=\ket{q_0}\!\bra{q_0}$. We have
\be
\ba
\eta^{q+1}\binom q{p-1}\sqrt{\binom{k+q+1}k\binom{l+q+1}l}&\lesseqgtr\eta^q\binom{q-1}{p-1}\sqrt{\binom{k+q}k\binom{l+q}l}\quad\quad\quad\quad\quad\quad\\
\Leftrightarrow\quad\eta&\lesseqgtr\frac{q-p+1}q\sqrt{\left(\frac{q+1}{k+q+1}\right)\left(\frac{q+1}{l+q+1}\right)}\\
\Leftrightarrow\quad\eta&\lesseqgtr\left(1-\frac{p-1}q\right)\sqrt{1-\frac k{k+q+1}}\sqrt{1-\frac l{l+q+1}},
\ea
\ee
where the right hand side in the last line is an increasing function of $q$. Hence, for $\eta\le\frac{p+1}{p\sqrt{(k+p+1)(l+p+1)}}$, which corresponds to $q=p$ in the above equation, we have
\be
\eta^{q+1}\binom q{p-1}\sqrt{\binom{k+q+1}k\binom{l+q+1}l}\le\eta^q\binom{q-1}{p-1}\sqrt{\binom{k+q}k\binom{l+q}l}
\ee
for all $q\ge p$, so that
\be
\max_{q\ge p}\left(\eta^q\binom{q-1}{p-1}\sqrt{\binom{k+q}k\binom{l+q}l}\right)=\eta^p\sqrt{\binom{k+p}k\binom{l+p}l},
\ee
for $\eta\le\frac{p+1}{p\sqrt{(k+p+1)(l+p+1)}}$.
With Eq.~(\ref{inductiong1}) we finally obtain
\be
\left|\underset{\alpha\leftarrow Q_\rho}{\mathbb E}[g_{k,l}^{(p)}(\alpha,\eta)]-\rho_{kl}\right|\le\eta^p\sqrt{\binom{k+p}k\binom{l+p}l}.
\label{boundEgkl}
\ee
\qed

\medskip

\noindent Note that for $\eta<1$ but $\eta>\frac{p+1}{p\sqrt{(k+p+1)(l+p+1)}}$, we can find the first integer $q_0$ such that $\eta\le\frac{(q_0-p+1)(q_0+1)}{q_0\sqrt{(k+q_0+1)(l+q_0+1)}}$ by incrementing $q$, since $\frac{(q-p+1)(q+1)}{q\sqrt{(k+q+1)(l+q+1)}}$ is an increasing function of $q$ which goes to $1$ when $q$ goes to $+\infty$. The bound above then gives
\be
\left|\underset{\alpha\leftarrow Q_\rho}{\mathbb E}[g_{k,l}^{(p)}(\alpha,\eta)]-\rho_{kl}\right|\le\eta^{q_0}\binom{q_0-1}{p-1}\sqrt{\binom{k+q_0}k\binom{l+q_0}l}.
\label{betterboundEgkl}
\ee
While this can be useful to optimise further the efficiency of the protocol, we will assume for the sake of obtaining closed analytical expressions that $\eta\le\frac{p+1}{p\sqrt{(k+p+1)(l+p+1)}}$.

%--------------------------------------------------------------------------------

\section{Proof of Theorem~\ref{th:fe}}
\label{app:thfe}

\noindent In this section, we prove the efficiency of Protocol~\ref{prot:fe}, which we recall below:

\begin{protocol}[Single-mode fidelity estimation]\label{prot:feapp}
Let $c\in\mathbb N^*$ and let $\ket C=\sum_{n=0}^{c-1}{c_n\ket n}$ be a core state. Let also $N,M\in\mathbb N^*$, and let $p\in\mathbb N^*$ and $0<\eta<1$ be free parameters. Let $\rho^{\otimes N+M}$ be $N+M$ copies of an unknown single-mode (mixed) quantum state $\rho$.
\begin{enumerate}
\item
Measure $N$ copies of $\rho$ with heterodyne detection, obtaining the samples $\alpha_1,\dots,\alpha_N\in\mathbb C$.
\item
Compute the mean $F_C(\rho)$ of the function $z\mapsto g_C^{(p)}(z,\eta)$ (defined in Eq.~(\ref{appgC})) over the samples $\alpha_1,\dots,\alpha_N\in\mathbb C$.
\item
Compute the fidelity estimate $F_C(\rho)^M$.
\end{enumerate}
\end{protocol}

\noindent The value $F_C(\rho)^M$ obtained is an estimate of the fidelity between the $M$ remaining copies of $\rho$ and $M$ copies of the target core state $\ket C$. Defining
\be
\eta_{C,p}:=\min\left\{\left[\frac{c\epsilon}{(c+p)\left(\sum_{n=0}^{c-1}{|c_n|\sqrt{\binom{n+p}n}}\right)^2}\right]^{\frac1p},\frac{p+1}{p(p+c)}\right\},
\label{etafeapp}
\ee
the efficiency of the protocol is summarised by the following result:

\begin{theo}\label{th:feapp}
Let $\epsilon>0$. With the notations of Protocol~\ref{prot:fe} and $\eta\le\eta_{C,p}$, we have
\be
\left|F(\ket C^{\otimes M},\rho^{\otimes M})-F_C(\rho)^M\right|\le\epsilon,
\ee
with probability greater than $1-P_C^{iid}$, where
\be
\ba
P_C^{iid}=2\exp\left[-\frac{N\epsilon^{2+\frac{2c}p}}{M^{2+\frac{2c}p}K_{C,p}}\right],
\ea
\ee
where $K_{C,p}$ is a constant independent of $\rho$ given in Eq.~(\ref{KCp}).
\end{theo}

\begin{proof}

\noindent Let us prove the theorem for $M=1$. The general case is then an immediate corollary by replacing $\epsilon$ by $\frac\epsilon M$, since $F(C^{\otimes M},\rho^{\otimes M})=F(C,\rho)^M$ and by Lemma~2 of~\cite{chabaud2019buildingarXiv}:
\be
|a-b|\le\epsilon\quad\Rightarrow\quad|a^M-b^M|\le M\epsilon,
\ee
for all $M\in\mathbb N^*$, all $\epsilon>0$ and all $a,b\in[0,1]$ (note that when the estimate $F_C(\rho)$ obtained is bigger than $1$ we may replace it by $1$ instead without loss of generality).

The proof combines Lemma~\ref{lem:main} with Hoeffding inequality. The latter is expressed as follows:

\begin{lem}\textbf{(Hoeffding)}\label{lem:Hoeffding} Let $\lambda>0$, let $n\ge1$, let $z_1,\dots,z_n$ be i.i.d.\ complex random variables from a probability density $D$ over $\mathbb R$, and let $g:\mathbb C\mapsto\mathbb R$ be a bounded function and let $G\ge\max_zg(z)-\min_zg(z)$. Then
\be
\Pr\left[\left|\frac1N\sum_{i=1}^N{g(z_i)}- \underset{z\leftarrow D}{\mathbb E}[g(z)]\right|\ge\lambda\right] \leq 2\exp\left[{-\frac{2N\lambda^2}{G^2}}\right].
\ee
\end{lem}

\noindent In order to apply Lemma~\ref{lem:Hoeffding}, we first derive an analytical bound for the functions $g_{k,l}^{(p)}$:

\begin{lem}\label{lem:boundgkl}
For all $k,l\in\mathbb N$, all $z\in\mathbb C$, all $p\in\mathbb N^*$ and all $\eta<1$,
\be
|g_{k,l}^{(p)}(z,\eta)|\le\frac1{\eta^{1+\frac{k+l}2}}\binom{\max{(k,l)}+p}{p-1}\sqrt{2^{|l-k|}\binom{\max{(k,l)}}{\min{(k,l)}}}.
\label{boundgkl}
\ee
\end{lem}

\begin{proof}

By Lemma~6 of \cite{chabaud2019buildingarXiv}, we have
\be
|f_{i,j}(z,\eta)|\le\frac1{\eta^{1+\frac{i+j}2}}\sqrt{2^{|i-j|}\binom{\max{(i,j)}}{\min{(i,j)}}},
\ee
for all $i,j\in\mathbb N$, all $z\in\mathbb C$ and all $\eta<1$.
Hence, with Eq.~(\ref{appg}), for all $k\le l\in\mathbb N$,
\be
\ba
|g_{k,l}^{(p)}(z,\eta)|&\le\sum_{i=0}^{p-1}{|f_{k+i,l+i}(z,\eta)|\eta^i\sqrt{\binom{k+i}i\binom{l+i}i}}\\
&\le\frac{\sqrt{2^{l-k}}}{\eta^{1+\frac{k+l}2}}\sum_{i=0}^{p-1}{\sqrt{\binom{l+i}{k+i}\binom{k+i}i\binom{l+i}i}}\\
&=\frac{\sqrt{2^{l-k}}}{\eta^{1+\frac{k+l}2}}\sum_{i=0}^{p-1}{\sqrt{\binom{l+i}i^2\binom lk}}\\
&=\frac1{\eta^{1+\frac{k+l}2}}\sqrt{2^{l-k}\binom lk}\sum_{i=0}^{p-1}{\binom{l+i}i}\\
&=\frac1{\eta^{1+\frac{k+l}2}}\binom{l+p}{p-1}\sqrt{2^{l-k}\binom lk},
\ea
\ee
for all $z\in\mathbb C$ and all $\eta<1$. The same reasoning holds for $l\le k\in\mathbb N$ and we obtain
\be
|g_{k,l}^{(p)}(z,\eta)|\le\frac1{\eta^{1+\frac{k+l}2}}\binom{\max{(k,l)}+p}{p-1}\sqrt{2^{|l-k|}\binom{\max{(k,l)}}{\min{(k,l)}}},
\ee
for all $k,l\in\mathbb N$, all $z\in\mathbb C$, all $p\in\mathbb N^*$ and all $\eta<1$.

\end{proof}

\noindent Let $\ket C=\sum_{n=0}^{c-1}{c_n\ket n}$ be a core state. Recalling the definition
\be
\ba
g_C^{(p)}(z,\eta)&:=\sum_{0\le k,l\le c-1}{c_k^*c_l\,g_{k,l}^{(p)}(z,\eta)},
\ea
\label{appgC}
\ee
from the main text, with Lemma~\ref{lem:boundgkl} we obtain:
\be
\ba
|g_C^{(p)}(z,\eta)|&\le\frac1{\eta^c}\sum_{0\le k,l\le c-1}{|c_kc_l|\eta^{c-1-\frac{k+l}2}\binom{\max{(k,l)}+p}{p-1}\sqrt{2^{|l-k|}\binom{\max{(k,l)}}{\min{(k,l)}}}}\\
&=\frac{B^{(p)}_C(\eta)}{\eta^c},
\ea
\label{boundgC}
\ee
for all $z\in\mathbb C$, all $p\in\mathbb N^*$ and all $\eta<1$, where
\be
B_C^{(p)}(\eta):=\sum_{0\le k,l\le c-1}{|c_kc_l|\eta^{c-1-\frac{k+l}2}\binom{\max(k,l)+p}{p-1}\sqrt{2^{|l-k|}\binom{\max(k,l)}{\min(k,l)}}}.
\label{BCp}
\ee
Let $N\in\mathbb N^*$, let $\rho=\sum_{k,l\ge0}{\rho_{kl}\ket k\!\bra l}$ be a density matrix and let $\alpha_1,\dots,\alpha_N$ be i.i.d.\ samples from balanced heterodyne detection of $N$ copies of $\rho$. Let us write
\be
F_C(\rho)=\frac1N\sum_{i=1}^N{g_C^{(p)}(\alpha_i,\eta)}.
\ee
With Lemma~\ref{lem:Hoeffding}, for all $\lambda>0$,
\be
\ba
\Pr\left[\left|F_C(\rho)-\underset{\alpha\leftarrow Q_\rho}{\mathbb E}[g_C^{(p)}(\alpha,\eta)]\right|\ge\lambda\right]&\le2\exp\left[{-\frac{2N\lambda^2\eta^{2c}}{\left(R^{(p)}_C(\eta)\right)^2}}\right]\\
&\le2\exp\left[{-\frac{N\lambda^2\eta^{2c}}{2\left(B^{(p)}_C(\eta)\right)^2}}\right],
\ea
\label{step1g}
\ee
for all $z\in\mathbb C$, all $p\in\mathbb N^*$, where we defined
\be
R^{(p)}_C(\eta):=\max_z{\left[\eta^{c}g_C(z,\eta)\right]}-\min_z{\left[\eta^{c}g_C(z,\eta)\right]},
\ee
the range of the function $z\mapsto \eta^{c}g_C(z,\eta)$ and where we used  Eq.~(\ref{boundgC}) and $R^{(p)}_C(\eta)\le2B^{(p)}_C(\eta)$ for $\eta<1$ in the second line. For $\eta\le\frac{p+1}{p(p+c)}$, we have $\eta\le\frac{p+1}{p\sqrt{(k+p+1)(l+p+1)}}$ for all $0\le k,l\le c-1$. Hence, for $\eta\le\frac{p+1}{p(p+c)}$,
\be
\ba
\left|F(C,\rho)-\underset{\alpha\leftarrow Q_\rho}{\mathbb E}[g_C^{(p)}(\alpha,\eta)]\right|&=\left|\sum_{0\le k,l\le c-1}{c_k^*c_l\left(\rho_{kl}-\underset{\alpha\leftarrow Q_\rho}{\mathbb E}[g_{k,l}^{(p)}(\alpha,\eta)]\right)}\right|\\
&\le\sum_{0\le k,l\le c-1}{|c_kc_l|\left|\rho_{kl}-\underset{\alpha\leftarrow Q_\rho}{\mathbb E}[g_{k,l}^{(p)}(\alpha,\eta)]\right|}\\
&\le\eta^p\sum_{0\le k,l\le c-1}{|c_kc_l|\sqrt{\binom{k+p}k\binom{l+p}l}}\\
&=\eta^p\left(\sum_{n=0}^{c-1}{|c_n|\sqrt{\binom{n+p}n}}\right)^2\\
&=\eta^pA_C^{(p)},
\ea
\label{step2g}
\ee
where we used Eq.~(\ref{appgC}) in the first line, Lemma~\ref{lem:main} in the third line for all $0\le k,l\le c-1$, and where we defined
\be
A_C^{(p)}:=\left(\sum_{n=0}^{c-1}{|c_n|\sqrt{\binom{n+p}n}}\right)^2.
\label{ACp}
\ee
For $\eta\le\frac{p+1}{p(p+c)}$, combining Eqs.~(\ref{step1g}) and (\ref{step2g}) yields
\be
\ba
\left|F(C,\rho)-F_C(\rho)\right|&\le\left|\underset{\alpha\leftarrow Q_\rho}{\mathbb E}[g_C^{(p)}(\alpha,\eta)]-F_C(\rho)\right|+\left|F(C,\rho)-\underset{\alpha\leftarrow Q_\rho}{\mathbb E}[g_C^{(p)}(\alpha,\eta)]\right|\\
&\le\lambda+\eta^pA_C^{(p)},
\ea
\ee
with probability greater than $1-P_C^{iid}$, where
\be
P_C^{iid}=2\exp\left[{-\frac{N\lambda^2\eta^{2c}}{2\left(B^{(p)}_C(\eta)\right)^2}}\right].
\label{Piid1}
\ee
We now optimise over the choice of $\lambda$ and $\eta$ for a given precision $\epsilon>0$. Namely, fixing $\lambda+\eta^pA_C^{(p)}=\epsilon$ we maximise $\lambda^2\eta^{2c}$. Writing
\be
\eta^pA_C^{(p)}=\mu,
\ee
this amounts to maximising $\lambda^2\mu^{\frac{2c}p}$, with $\lambda+\mu=\epsilon$. The optimal choice is given by:
\be
\lambda=\frac{p\epsilon}{c+p}\quad\text{and}\quad\mu=\frac{c\epsilon}{c+p}.
\ee
This in turn gives
\be
\lambda=\frac{p\epsilon}{c+p}\quad\text{and}\quad\eta=\left[\frac{c\epsilon}{(c+p)A_C^{(p)}}\right]^{\frac1p}.
\ee
With this choice for $\eta$, for small enough $\epsilon$ we will always have $\eta\le\frac{p+1}{p(p+c)}$. To ensure that this is always the case independently of the value of $\epsilon$, we may fix instead
\be
\lambda=\frac{p\epsilon}{c+p}\quad\text{and}\quad\eta=\min\left\{\left[\frac{c\epsilon}{(c+p)A_C^{(p)}}\right]^{\frac1p},\frac{p+1}{p(p+c)}\right\}.
\ee
Then, plugging these expressions in Eq.~(\ref{Piid1}) we finally obtain
\be
\left|F(C,\rho)-F_C(\rho)\right|\le\epsilon,
\ee
with probability greater than $1-P_C^{iid}$, where
\be
\ba
P_C^{iid}&=\max\left\{2\exp\left[-\frac{N\epsilon^{2+\frac{2c}p}}{K_C^{(p)}(\eta)}\right],2\exp\left[-\frac{N\epsilon^2}{K_C^{(p)'(\eta)}}\right]\right\}\\
&=2\exp\left[-\frac{N\epsilon^{2+\frac{2c}p}}{K_{C,p}(\eta)}\right],
\ea
\label{Piid2}
\ee
where
\be
\ba
K_C^{(p)}(\eta)&=\frac{2\left(A_C^{(p)}\right)^{\frac{2c}p}\left(B_C^{(p)}(\eta)\right)^2(c+p)^{2+\frac{2c}p}}{c^{\frac{2c}p}p^2}\\
K_C^{(p)'}(\eta)&=\frac{2p^{2c-2}\left(B_C^{(p)}(\eta)\right)^2(c+p)^{2c+2}}{(p+1)^{2c}}\\
K_{C,p}(\eta)&=\max\left\{K_C^{(p)}(\eta),K_C^{(p)'}(\eta)\right\},
\ea
\label{KCp}
\ee
and where the constants $B_C^{(p)}(\eta)$ and $A_C^{(p)}$ are defined in Eqs.~(\ref{BCp}) and (\ref{ACp}), respectively.
Hence, the number of samples needed for a precision $\epsilon>0$ and a failure probability $\delta>0$ scales as
\be
N_1=\mathcal O\left(\frac1{\epsilon^{2+\frac{2c}p}}\log\left(\frac1\delta\right)\right).
\label{Nfeapp}
\ee

\end{proof}

\noindent Note that, as indicated in Eq.~(\ref{betterboundEgkl}), the choice of $\eta$ in Eq.~(\ref{etafeapp}) is not optimal in general, and an optimisation depending on the expression of the target core state leads to a better prefactor in Eq.~(\ref{Nfeapp}).

%--------------------------------------------------------------------------------

\section{Removing the i.i.d.\ assumption for Protocol~\ref{prot:fe}}
\label{app:fenotiid}

\noindent In this section, we derive a version of Protocol~\ref{prot:fe} which does not assume i.i.d.\ state preparation. The protocol is as follows:

\begin{prot}[General single-mode fidelity estimation]\label{prot:fenotiid}
Let $c\in\mathbb N^*$ and let $\ket C=\sum_{n=0}^{c-1}{c_n\ket n}$ be a core state. Let also $N,M\in\mathbb N^*$, and let $p\in\mathbb N^*$, $E,S\in\mathbb N$ and $0<\eta<1$ be free parameters. Let $\bm\rho^{N+M}$ be an unknown quantum state over $N+M$ subsystems. We write $N=N'+K+Q$, for $N',K,Q\in\mathbb N$.
\begin{enumerate}
\item
Measure $N=N'+K+Q$ subsystems of $\bm\rho^{N+M}$ chosen at random with heterodyne detection, obtaining the samples $\alpha_1,\dots,\alpha_{N'},\beta_1,\dots,\beta_K,\gamma_1,\dots\gamma_Q\in\mathbb C$. Let $\bm\rho^M$ be the remaining state over $M$ subsystems.
\item
Discard the $Q$ samples $\gamma_1,\dots\gamma_Q$.
\item
Record the number $R$ of samples $\beta_i$ such that $|\beta_i|^2+1>E$, for $i\in\{1,\dots,K\}$. The protocol aborts if $R>S$.
\item
Otherwise, compute the mean $F_C(\bm\rho)$ of the function $z\mapsto g_C^{(p)}(z,\eta)$ (defined in Eq.~(\ref{gCmain})) over the samples $\alpha_1,\dots,\alpha_{N'}\in\mathbb C$.
\item
Compute the fidelity estimate $F_C(\bm\rho)^M$.
\end{enumerate}
\end{prot}

\noindent Note that this protocol differs from Protocol~\ref{prot:fe} by two additional classical post-processing steps, steps 2 and 3. These steps are part of a de Finetti reduction for infinite-dimensional systems~\cite{renner2009finetti} detailed in~\cite{chabaud2019buildingarXiv}. In particular, step 3 is an energy test, and the energy parameters $E$ and $S$ should be chosen to guarantee completeness, i.e., if the perfect core state is sent then it passes the energy test with high probability. 

The value $F_C(\bm\rho)^M$ obtained is an estimate of the fidelity between the remaining state $\bm\rho^M$ over $M$ subsystems and $M$ copies of the target core state $\ket C$. The efficiency of the protocol is summarised by the following result:

\begin{theo}\label{th:fenotiid}
Let $\epsilon>0$. With the notations of Protocol~\ref{prot:fenotiid} and $\eta\le\eta_{C,p}$, defined in Eq.~(\ref{etafeapp}),
\be
\left|F(\ket C^{\otimes M},\bm\rho^M)-F_C(\bm\rho)^M\right|\le2\epsilon+P_C^{\text{deFinetti}},
\ee
or the protocol aborts in step 3, with probability greater than $1-(P_C^{\text{support}}+P_C^{\text{deFinetti}}+P_C^{\text{choice}}+P_C^{\text{Hoeffding}})$, where
\be
\ba
P_C^{\text{support}}&=8K^{3/2}\exp\left[-\frac K9\left(\frac Q{4(N'+M+Q)}-\frac{2S}K\right)^2\right]\\
P_C^{\text{deFinetti}}&=\left(\frac Q4\right)^{E^2/2}\exp\left[-\frac{Q(Q+4)}{8(N'+M+Q)}\right]\\
P_C^{\text{choice}}&=\frac{M(Q+M-1)}{N'+M}\\
P_C^{\text{Hoeffding}}&=2\binom{N'+M}Q\exp\left[-\frac{N'+M-Q}{2M^{2+\frac{2c}p}}\left(\frac{\epsilon^{1+\frac cp}}{G_{C,p}}-\frac{2QM^{1+\frac cp}}{N'}\right)^2\right],
\ea
\ee
where $G_{C,p}$ is a constant independent of $\rho$ given in Eq.~(\ref{GCp}). In particular, setting
\be
\ba
N'&=\mathcal O\left(\frac{M^{7+\frac{4c}p}}{\epsilon^{4+\frac{4c}p}}\right)\\
Q&=\mathcal O\left(\frac{M^{4+\frac{2c}p}}{\epsilon^{2+\frac{2c}p}}\right)\\
K&=\mathcal O\left(\frac{M^{7+\frac{4c}p}}{\epsilon^{4+\frac{4c}p}}\right)\\
E&=\mathcal O(\log M)\\
S&=o(Q),
\ea
\ee
yields
\be
\left|F(\ket C^{\otimes M},\bm\rho^M)-F_C(\bm\rho)^M\right|=\mathcal O(\epsilon),
\ee
when $\frac1\epsilon=\mathcal O(\poly M)$, or the protocol aborts in step 3, with probability greater than $1-P_C$,
where
\be
P_C=\mathcal O\left(\frac1{\poly(M,\frac1\epsilon)}\right),
\ee
for
\be
N=\mathcal O\left(\frac{M^{7+\frac{4c}p}}{\epsilon^{4+\frac{4c}p}}\right).
\ee
\end{theo}

\begin{proof}

\noindent Theorem~\ref{th:fenotiid} is an optimised version of Theorem~5 of~\cite{chabaud2019buildingarXiv}, where the heterodyne estimates $f_{k,l}(z,\eta)$ defined in Eq.~(\ref{fkl}) are replaced by the more efficient generalised estimates $g_{k,l}^{(p)}(z,\eta)$, with an additional free parameter $p$. The proof thus is nearly identical, and amounts to replacing occurences of the estimates $f_{k,l}(z,\eta)$ by the estimates $g_{k,l}^{(p)}(z,\eta)$. In particular, we use our Eq.~(\ref{step2g}) instead of Theorem~1 of~\cite{chabaud2019buildingarXiv}. For a target core state $\ket\psi=\ket C$, this amounts to replace
\be
\left|F(\psi,\rho)-\underset{\alpha\leftarrow Q_\rho}{\mathbb E}[f_\psi(\alpha,\eta)]\right|\le\eta K_\psi
\ee
by
\be
\left|F(C,\rho)-\underset{\alpha\leftarrow Q_\rho}{\mathbb E}[g_C^{(p)}(\alpha,\eta)]\right|\le\eta^pA_C^{(p)},
\ee
where $A^{(p)}_C$ is defined in Eq.~(\ref{ACp}). Moreover, we use our Eq.~(\ref{boundgC}) instead of Lemma~6 of~\cite{chabaud2019buildingarXiv}. For a target core state $\ket\psi=\ket C$, this amounts to replace
\be
|f_\psi(z,\eta)|\le\frac{M_\psi(\eta)}{\eta^c}
\ee
by
\be
|g_C^{(p)}(z,\eta)|\le\frac{B^{(p)}_C}{\eta^c},
\ee
where $B^{(p)}_C$ is defined in Eq.~(\ref{BCp}). Hence, the constants used in the proof of~\cite{chabaud2019buildingarXiv} are mapped as:
\be\label{mappingKM}
\ba
K_\psi&\mapsto\eta^{p-1}A_C^{(p)}\\
M_\psi(\eta)&\mapsto B_C^{(p)}(\eta).
\ea
\ee
Note that we also adapt the notations from~\cite{chabaud2019buildingarXiv} for the global coherence of this work as:
\be
\ba
n&\mapsto N'+M+Q\\
k&\mapsto K\\
q&\mapsto Q/4\\
m&\mapsto M\\
E&\mapsto E-1\\
s&\mapsto S\\
r&\mapsto R\\
\epsilon&\mapsto\epsilon\\
\epsilon'&\mapsto\epsilon\\
\ket\psi&\mapsto\ket C,
\ea
\ee
and we write $c$ the support size of $\ket\psi=\ket C$ instead of $E+1$\footnote{Ref.~\cite{chabaud2019buildingarXiv} uses the letter $E$ both for the support size of the target state $\ket\psi$ and the energy test threshold, while we keep these notations different.}. Ref.~\cite{chabaud2019buildingarXiv} sets $\eta=\frac\epsilon{mK_\psi}$ and defines for brevity
\be\label{Cpsi}
C_\psi:=K_\psi^{E+1}M_\psi\left(\frac\epsilon{mK_\psi}\right).
\ee
With $m\mapsto M$, we thus set $\eta=\frac\epsilon{M\eta^{p-1}A_C^{(p)}}$,  or equivalently
\be\label{seteta}
\eta=\left[\frac\epsilon{MA_C^{(p)}}\right]^{\frac1p}.
\ee
By Eq.~(\ref{mappingKM}), $K_\psi$ is mapped to $\left(\frac\epsilon M\right)^{1-\frac1p}\left[A_C^{(p)}\right]^{\frac1p}$, and the expression $C_\psi$ in Eq.~(\ref{Cpsi}) is thus replaced here by
\be\label{GCp}
G_{C,p}:=\left(\frac\epsilon M\right)^{c-\frac cp}\left[A_C^{(p)}\right]^{\frac cp}B_C^{(p)}\left(\left[\frac\epsilon{MA_C^{(p)}}\right]^{\frac1p}\right).
\ee
This expression is independent of $\rho$ and depends on $M$ and $\epsilon$ without loss of generality, since is upper bounded by its value for $\epsilon/M=1$.
With these mappings, following the proof of~\cite{chabaud2019buildingarXiv} we obtain:
\be
\left|F(\ket C^{\otimes M},\bm\rho^M)-F_C(\bm\rho)^M\right|\le2\epsilon+P_C^{deFinetti},
\label{estimateFnotiid}
\ee
or $R>S$, i.e., the protocol aborts in step 3, with probability greater than $1-(P_C^{support}+P_C^{deFinetti}+P_C^{choice}+P_C^{Hoeffding})$, where
\be
\ba
P_C^{\text{support}}&=8K^{3/2}\exp\left[-\frac K9\left(\frac Q{4(N'+M+Q)}-\frac{2S}K\right)^2\right]\\
P_C^{\text{deFinetti}}&=\left(\frac Q4\right)^{E^2/2}\exp\left[-\frac{Q(Q+4)}{8(N'+M+Q)}\right]\\
P_C^{\text{choice}}&=\frac{M(Q+M-1)}{N'+M}\\
P_C^{\text{Hoeffding}}&=2\binom{N'+M}Q\exp\left[-\frac{N'+M-Q}{2M^{2+\frac{2c}p}}\left(\frac{\epsilon^{1+\frac cp}}{G_{C,p}}-\frac{2QM^{1+\frac cp}}{N'}\right)^2\right],
\ea
\ee
where $G_{C,p}$ is defined in Eq.~(\ref{GCp}). Setting for simplicity $K=\mathcal O(N')$ and $N'=qQ$, we find that all above probabilities are at least inverse polynomially small in $M$ and $\frac1\epsilon$ when
\be
\left(\frac M\epsilon\right)^{2+\frac{2c}p}<<q<<Q,
\ee
and $E=\mathcal O(\log M)$ and $S=o(Q)$. In particular, setting
\be
\ba
N'&=\mathcal O\left(\frac{M^{7+\frac{4c}p}}{\epsilon^{4+\frac{4c}p}}\right)\\
Q&=\mathcal O\left(\frac{M^{4+\frac{2c}p}}{\epsilon^{2+\frac{2c}p}}\right)\\
K&=\mathcal O\left(\frac{M^{7+\frac{4c}p}}{\epsilon^{4+\frac{4c}p}}\right)\\
\ea
\ee
with Eq.~(\ref{estimateFnotiid}) we obtain
\be\label{estimateFnotiidbis}
\ba
\left|F(\ket C^{\otimes M},\bm\rho^M)-F_C(\bm\rho)^M\right|&=2\epsilon+e^{-\mathcal O(M)+O(\log^2M\log\frac M\epsilon)}\\
&=\mathcal O(\epsilon),
\ea
\ee
when $\frac1\epsilon=\mathcal O(\poly M)$, with probability $1-P_C$, where
\be
P_C=P_C^{\text{support}}+P_C^{\text{deFinetti}}+P_C^{\text{choice}}+P_C^{\text{Hoeffding}}=\mathcal O\left(\frac1{\poly(M,\frac1\epsilon)}\right),
\ee
for
\be
N=N'+K+Q=\mathcal O\left(\frac{M^{7+\frac{4c}p}}{\epsilon^{4+\frac{4c}p}}\right).
\ee

\end{proof}

\noindent Since the parameter $p\in\mathbb N^*$ may vary freely, the asymptotic scaling of the number of samples $N$ needed for a precision $\epsilon>0$ and a polynomial confidence thus is given by $N=\mathcal O(\frac{M^7}{\epsilon^4})$. Moreover, we set a choice for $\eta$ in Eq.~(\ref{seteta}) for simplicity, but the choice of both $p$ and $\eta$ can further improve the efficiency of the protocol.

%--------------------------------------------------------------------------------

\section{Proof of Lemma~\ref{lem:product}}
\label{app:lemproduct}

\noindent In this section, we prove a generalised version of Lemma~\ref{lem:product} for tensor product of Hilbert spaces. Setting $M=1$, we retrieve the Lemma in the main text:
 
\begin{lem}\label{lem:productmulti}
Let $\bm\rho^{m\times M}$ be an $m\times M$-mode state in $\mathcal H=\mathcal H_1^{\otimes M}\otimes\dots\otimes\mathcal H_m^{\otimes M}$. For all $i\in\{1,\dots,m\}$, we denote by $\bm\rho_i^M=\Tr_{\mathcal H\setminus\mathcal H_i^{\otimes M}}(\bm\rho)$ the $M$-mode reduced state of $\bm\rho^{m\times M}$ over the Hilbert space $\mathcal H_i^{\otimes M}$. Let $\ket{\psi_1},\dots,\ket{\psi_m}$ be pure states in $\mathcal H_1,\dots,\mathcal H_m$. Then,
\be
1-m\left(1-F\left(\bm\rho^{m\times M},\bigotimes_{i=1}^m\psi_i^{\otimes M}\right)\right)\le1-\sum_{i=1}^m{\left(1-F\left(\bm\rho_i^M,\psi_i^{\otimes M}\right)\right)}\le F\left(\bm\rho^{m\times M},\bigotimes_{i=1}^m\psi_i^{\otimes M}\right).
\label{fideproductmult}
\ee
\end{lem}

\begin{proof}

\noindent Since $\ket{\psi_1},\dots,\ket{\psi_m}$ are pure states,
\be
F\left(\bm\rho^{m\times M},\bigotimes_{i=1}^m\psi_i^{\otimes M}\right)=\Tr\,[\bm\rho^{m\times M}\ket{\psi_1}\!\bra{\psi_1}^{\otimes M}\otimes\dots\otimes\ket{\psi_m}\!\bra{\psi_m}^{\otimes M}],
\label{fidepure1mult}
\ee
and
\be
\ba
F(\bm\rho_i^M,\psi_i^{\otimes M})&=\Tr[\bm\rho_i^M\ket{\psi_i}\!\bra{\psi_i}^{\otimes M}]\\
&=\Tr[\bm\rho^{m\times M}(\mathbb1\otimes\ket{\psi_i}\!\bra{\psi_i}^{\otimes M}\otimes\mathbb1)],
\ea
\label{fidepure2mult}
\ee
for all $i\in\{1,\dots,m\}$, where $\mathbb1$ is the identity operator.
The right hand side of Eq.~(\ref{fideproductmult}) is obtained by writing $F(\bm\rho^{m\times M},\bigotimes_{i=1}^m\psi_i^{\otimes M})$ as a telescopic sum:
\be
\ba
\Tr\,[\bm\rho^{m\times M}\ket{\psi_1}\!\bra{\psi_1}^{\otimes M}\otimes\dots&\otimes\ket{\psi_m}\!\bra{\psi_m}^{\otimes M}]=\Tr\,[\bm\rho^{m\times M}\mathbb1]\quad\quad\quad\\
&-\Tr\,[\bm\rho^{m\times M}(\mathbb1-\ket{\psi_1}\!\bra{\psi_1}^{\otimes M})\otimes\mathbb1]\quad\quad\quad\\
&-\Tr\,[\bm\rho^{m\times M}\ket{\psi_1}\!\bra{\psi_1}^{\otimes M}\otimes(\mathbb1-\ket{\psi_2}\!\bra{\psi_2}^{\otimes M})\otimes\mathbb1]\quad\quad\quad\\
&-\Tr\,[\bm\rho^{m\times M}\ket{\psi_1}\!\bra{\psi_1}^{\otimes M}\otimes\ket{\psi_2}\!\bra{\psi_2}^{\otimes M}\otimes(\mathbb1-\ket{\psi_3}\!\bra{\psi_3}^{\otimes M})\otimes\mathbb1]\quad\quad\quad\\
&-\dots\quad\quad\quad\\
&-\Tr\,[\bm\rho^{m\times M}\ket{\psi_1}\!\bra{\psi_1}^{\otimes M}\otimes\ket{\psi_2}\!\bra{\psi_2}^{\otimes M}\otimes\dots\otimes(\mathbb1-\ket{\psi_m}\!\bra{\psi_m}^{\otimes M})]\quad\quad\quad\\
&\ge1-\sum_{i=1}^m{\Tr[\bm\rho^{m\times M}(\mathbb1\otimes(\mathbb1-\ket{\psi_i}\!\bra{\psi_i}^{\otimes M})\otimes\mathbb1)]}\\
&=1-\sum_{i=1}^m{\left(1-\Tr[\bm\rho^{m\times M}(\mathbb1\otimes\ket{\psi_i}\!\bra{\psi_i}^{\otimes M}\otimes\mathbb1)]\right)},
\ea
\ee
by linearity of the trace, where we used $\Tr\,(\bm\rho^{m\times M})=1$. This gives
\be
F\left(\bm\rho^{m\times M},\bigotimes_{i=1}^m\psi_i^{\otimes M}\right)\ge1-\sum_{i=1}^m{\left(1-F\left(\bm\rho_i^M,\psi_i^{\otimes M}\right)\right)},
\ee
with Eqs.~(\ref{fidepure1mult}) and (\ref{fidepure2mult}).

The left hand side of Eq.~(\ref{fideproductmult}) is obtained using:
\be
\ba
F\left(\bm\rho^{m\times M},\bigotimes_{i=1}^m\psi_i^{\otimes M}\right)&=\Tr\,[\bm\rho^{m\times M}\ket{\psi_1}\!\bra{\psi_1}^{\otimes M}\otimes\dots\otimes\ket{\psi_m}\!\bra{\psi_m}^{\otimes M}]\\
&\le\Tr\,[\bm\rho^{m\times M}(\mathbb1\otimes\ket{\psi_i}\!\bra{\psi_i}^{\otimes M}\otimes\mathbb1)]\\
&=F\left(\bm\rho_i^M,\psi_i^{\otimes M}\right),
\ea
\ee
for all $i\in\{1,\dots,m\}$, and summing over $i$. 

\end{proof}

\noindent Note that this result is not restricted to a particular protocol for fidelity estimation nor to the continuous-variable regime: if the fidelities of the single subsystems of a quantum state over $m$ subsystems with some target pure states are higher than $1-\frac\lambda m$, for some $\lambda>0$, then by Lemma~\ref{lem:product} the general state has fidelity at least $1-\lambda$ with the tensor product of the $m$ pure states.

%--------------------------------------------------------------------------------

\section{Proof of Lemma~\ref{lem:hetmagic}}
\label{app:lemhetmagic}

\noindent In this section, we prove Lemma~\ref{lem:hetmagic} from the main text which we recall below:

\begin{lem}\label{lemma:hetmagic} Let $\bm\beta,\bm\xi\in\mathbb C^m$ and let $\hat V=\hat S(\bm\xi)\hat D(\bm\beta)\,\hat U$, where $\hat U$ is an $m$-mode passive linear transformation with $m\times m$ unitary matrix $U$. For all $\bm\alpha\in\mathbb C^m$, let $\bm\gamma=U\bm\alpha+\bm\beta$. Then,
\be
\Pi^{\bm\xi}_{\bm\gamma}=\hat V\Pi^{\bm0}_{\bm\alpha}\hat V^\dag.
\ee
\end{lem}

\begin{proof}

Recall that the POVM elements for a tensor product of unbalanced heterodyne detections with unbalancing parameters $\bm\xi\in\mathbb C^m$ are given by
\be
\Pi_{\bm\alpha}^{\bm\xi}=\frac1{\pi^m}\ket{\bm\alpha,\bm\xi}\!\bra{\bm\alpha,\bm\xi},
\ee
for all $\bm\alpha,\bm\xi\in\mathbb C^m$, where $\ket{\bm\alpha,\bm\xi}=\hat S(\bm\xi)\hat D(\bm\alpha)\ket{\bm0}$ is a tensor product of squeezed coherent states. 

The POVM elements of a tensor product of single-mode balanced heterodyne detection are given by $\Pi_{\bm\alpha}^{\bm0}$, for all $\bm\alpha\in\mathbb C^m$, and we have
\be
\Pi_{\bm\alpha}^{\bm\xi}=\hat S(\bm\xi)\Pi_{\bm\alpha}^{\bm0}\hat S^\dag(\bm\xi).
\label{het1}
\ee
In particular, a single-mode squeezing followed by a single-mode balanced heterodyne detection can be simulated by performing directly an heterodyne detection unbalanced according to the squeezing parameter. One retrieves balanced heterodyne detection by setting the unbalancing parameter to $0$ and homodyne detection by letting the modulus of the unbalancing parameter go to infinity.

Passive linear transformations correspond to unitary transformations of the creation and annihilation operators of the modes. These transformations, which may be implemented by unitary optical interferometers, map coherent states to coherent states: if $\hat U$ is a passive linear transformation and $U$ is the unitary matrix describing its action on the creation and annihilation operators of the modes, an input coherent state $\ket{\bm\alpha}$ is mapped to an output coherent state $\hat U\ket{\bm\alpha}=\ket{U\bm\alpha}$, where $U\bm\alpha$ is obtained by multiplying the vector $\bm\alpha$ by the unitary matrix $U$. Hence, the POVM elements of a passive linear transformation followed by a product of single-mode balanced heterodyne detections are given by
\be
\hat U\Pi^{\bm0}_{\bm\alpha}\hat U^\dag=\Pi^{\bm0}_{U\bm\alpha},
\label{het2}
\ee
for all $\bm\alpha\in\mathbb C^m$. This implies that the passive linear transformation $\hat U^\dag$ followed by a tensor product of single-mode balanced heterodyne detections can be simulated by performing the heterodyne detections first, then multiplying the vector of samples obtained by $U$.

A similar property holds with single-mode displacements: since displacements map coherent states to coherent states, up to a global phase, by displacing their amplitude, a single-mode displacement followed by a single-mode heterodyne detection can be simulated by performing the heterodyne detection first, then translating the sample obtained according to the displacement amplitude. In particular we have
\be
\hat D(\bm\beta)\Pi_{\bm\alpha}^{\bm0}\hat D^\dag(\bm\beta)=\Pi_{\bm\alpha+\bm\beta}^{\bm0},
\label{het3}
\ee
for all $\bm\alpha,\bm\beta\in\mathbb C^m$, where $\bm\alpha+\bm\beta=(\alpha_1+\beta_1,\dots,\alpha_m+\beta_m)$.

Combining the properties of heterodyne detection in Eqs.~(\ref{het1}), (\ref{het2}) and (\ref{het3}), we have:
\be
\begin{cases}
\hat U\Pi^{\bm0}_{\bm\alpha}\hat U^\dag=\Pi^{\bm0}_{U\bm\alpha},\\
\hat D(\bm\beta)\Pi_{\bm\alpha}^{\bm0}\hat D^\dag(\bm\beta)=\Pi_{\bm\alpha+\bm\beta}^{\bm0},\\
\hat S(\bm\xi)\Pi_{\bm\alpha}^{\bm0}\hat S^\dag(\bm\xi)=\Pi_{\bm\alpha}^{\bm\xi},
\end{cases}
\ee
for all $\bm\alpha,\bm\beta,\bm\xi\in\mathbb C^m$ and all $m$-mode passive linear transformations $\hat U$ with $m\times m$ unitary matrix $U$. Writing $\hat V=\hat S(\bm\xi)\hat D(\bm\beta)\hat U$, we obtain
\be
\ba
\hat V\Pi_{\bm\alpha}^{\bm0}\hat V^\dag&=\hat S(\bm\xi)\hat D(\bm\beta)\hat U\Pi_{\bm\alpha}^{\bm0}\hat U^\dag\hat D^\dag(\bm\beta)\hat S^\dag(\bm\xi)\\
&=\hat S(\bm\xi)\hat D(\bm\beta)\Pi_{U\bm\alpha}^{\bm0}\hat D^\dag(\bm\beta)\hat S^\dag(\bm\xi)\\
&=\hat S(\bm\xi)\Pi_{U\bm\alpha+\bm\beta}^{\bm0}\hat S^\dag(\bm\xi)\\
&=\Pi_{U\bm\alpha+\bm\beta}^{\bm\xi}.
\ea
\ee

\end{proof}

%--------------------------------------------------------------------------------

\section{Proof of Theorem~\ref{th:we}}
\label{app:thwe}

\noindent In this section, we prove the efficiency of Protocol~\ref{prot:we}, which we recall below:

\begin{protocol}[Multimode fidelity witness estimation]\label{prot:weapp}
Let $c_1,\dots,c_m\in\mathbb N^*$. Let $\ket{C_i}=\sum_{n=0}^{c_i-1}{c_{i,n}\ket n}$ be a core state, for all $i\in\{1,\dots,m\}$. Let $\hat U$ be an $m$-mode passive linear transformation with $m\times m$ unitary matrix $U$, and let $\bm\beta,\bm\xi\in\mathbb C^m$. We write $\ket{\bm\psi}=\hat S(\bm\xi)\hat D(\bm\beta)\,\hat U\bigotimes_{i=1}^m\ket{C_i}$ the $m$-mode target pure state. Let $N,M\in\mathbb N^*$, and let $p_1,\dots,p_m\in\mathbb N^*$ and $0<\eta_1,\dots,\eta_m<1$ be free parameters. Let $\bm\rho^{\otimes N+M}$ be $N+M$ copies of an unknown $m$-mode (mixed) quantum state $\bm\rho$.
\begin{enumerate}
\item
Measure all $m$ subsystems of $N$ copies of $\bm\rho$ with unbalanced heterodyne detection with unbalancing parameters $\bm\xi=(\xi_1,\dots,\xi_m)$, obtaining the vectors of samples $\bm\gamma^{(1)},\dots,\bm\gamma^{(N)}\in\mathbb C^m$.
\item
For all $k\in\{1,\dots,N\}$, compute the vectors $\bm\alpha^{(k)}=U^\dag\left(\bm\gamma^{(k)}-\bm\beta\right)$. We write $\bm\alpha^{(k)}=(\alpha^{(k)}_1,\dots\alpha^{(k)}_m)\in\mathbb C^m$.
\item
For all $i\in\{1,\dots,m\}$, compute the mean $F_{C_i}(\bm\rho)$ of the function $z\mapsto g_{C_i}^{(p_i)}(z,\eta_i)$ (defined in Eq.~(\ref{gCmain})) over the samples $\alpha_i^{(1)},\dots,\alpha_i^{(N)}\in\mathbb C$.
\item
Compute the fidelity witness estimate $W_{\bm\psi}=1-\sum_{i=1}^m{(1-F_{C_i}(\bm\rho)^M)}$.
\end{enumerate}
\end{protocol}

\noindent The value $W_{\bm\psi}$ obtained is an estimate of a tight lower bound on the fidelity between the $M$ remaining copies of $\bm\rho$ and $M$ copies of the $m$-mode target state $\ket{\bm\psi}=\hat S(\bm\xi)\hat D(\bm\beta)\,\hat U\bigotimes_{i=1}^m\ket{C_i}$. The efficiency and completeness of the protocol are summarised by the following result:

\begin{theo}\label{th:weapp}
Let $\epsilon>0$. With the notations of Protocol~\ref{prot:weapp},
\be
1-m\left[1-F(\ket{\bm\psi}^{\otimes M},\bm\rho^{\otimes M})\right]-\epsilon\le W_{\bm\psi}\le F(\ket{\bm\psi}^{\otimes M},\bm\rho^{\otimes M})+\epsilon,
\ee
with probability $1-P^{iid}_W$, where
\be
P^{iid}_W=2\sum_{i=1}^m{\exp\left[-\frac{N\epsilon^{2+\frac{2c_i}{p_i}}}{M^{2+\frac{2c_i}{p_i}}K_{C_i,p_i}}\right]},
\label{PiidW}
\ee
where for all $i\in\{1,\dots,m\}$, $K_{C_i,p_i}$ is a constant independent of $\bm\rho$, defined in Eq.~(\ref{KCp}).
Moreover, if the tested state is perfect, i.e., $\rho^{\otimes N+M}=\ket\psi\!\bra\psi^{\otimes N+M}$, then $W_{\bm\psi}\ge1-\epsilon$ with probability $1-P^{iid}_W$.
\end{theo}

\begin{proof}

\noindent We use the notations of the protocol above and we first prove the theorem in the case where the target state is a tensor product of single-mode pure states, i.e., $\hat S(\bm\xi)\hat D(\bm\beta)\,\hat U=\mathbb1$ and $\bm\alpha^{(k)}=\bm\gamma^{(k)}$ for all $k\in\{1,\dots,N\}$. 

We thus consider for the target state a tensor product $\ket{\bm\psi}=\bigotimes_{i=1}^m\ket{C_i}$ of core states. For all $i\in\{1,\dots,m\}$ we write $\ket{C_i}=\sum_{n=0}^{c_i-1}{c_{i,n}\ket n}$, where $c_i\in\mathbb N^*$, and $\rho_i=\Tr_{\{1,\dots,m\}\setminus\{i\}}(\bm\rho)$ the reduced state of $\bm\rho$ over the $i^{th}$ mode. Let $\epsilon>0$. For all $i\in\{1,\dots,m\}$, with $\eta_i\le\eta_{C_i,p_i}$, defined in Eq.~(\ref{etafeapp}), we have by Theorem~\ref{th:feapp}, with $\frac\epsilon m$:
\be
\left|F(\ket{C_i}^{\otimes M},\rho_i^{\otimes M})-F_{C_i}(\bm\rho)^M\right|\le\frac\epsilon m,
\ee
with probability higher than $1-P_{C_i}^{iid}$, where
\be
P_{C_i}^{iid}=2\exp\left[-\frac{N\epsilon^{2+\frac{2c_i}{p_i}}}{M^{2+\frac{2c_i}{p_i}}m^{2+\frac{2c_i}{p_i}}K_{C_i,p_i}}\right],
\ee
where $K_{C_i,p_i}$ is a constant independent of $\bm\rho$ given in Eq.~(\ref{KCp}). Hence, taking the union bound of the failure probabilities $P_{C_i}^{iid}$, we obtain
\be
\ba
\left|1-\sum_{i=1}^m{\left(1-F(\ket{C_i}^{\otimes M},\rho_i^{\otimes M})\right)}-W_{\bm\psi}\right|&=\left|\sum_{i=1}^m{\left(F(\ket{C_i}^{\otimes M},\rho_i^{\otimes M})-F_{C_i}(\bm\rho)^M\right)}\right|\\
&\le\sum_{i=1}^m\left|F(\ket{C_i}^{\otimes M},\rho_i^{\otimes M})-F_{C_i}(\bm\rho)^M\right|\\
&\le\epsilon,
\ea
\label{unionboundW}
\ee
with probability higher than $1-P_W^{iid}$, where
\be
P_W^{iid}=2\sum_{i=1}^m{\exp\left[-\frac{N\epsilon^{2+\frac{2c_i}{p_i}}}{M^{2+\frac{2c_i}{p_i}}m^{2+\frac{2c_i}{p_i}}K_{C_i,p_i}}\right]}.
\ee
By Lemma~\ref{lem:productmulti}, with $\bm\rho^{m\times M}=\bm\rho^{\otimes M}$ and $\ket{\psi_i}=\ket{C_i}$,
\be
1-m\left[1-F(\ket{\bm\psi}^{\otimes M},\bm\rho^{\otimes M})\right]\le 1-\sum_{i=1}^m{\left(1-F(\ket{C_i}^{\otimes M},\rho_i^{\otimes M})\right)}\le F(\ket{\bm\psi}^{\otimes M},\bm\rho^{\otimes M}).
\ee
With Eq.~(\ref{unionboundW}) we obtain
\be
1-m\left[1-F(\ket{\bm\psi}^{\otimes M},\bm\rho^{\otimes M})\right]-\epsilon\le W_{\bm\psi}\le F(\ket{\bm\psi}^{\otimes M},\bm\rho^{\otimes M})+\epsilon,
\label{estimateWapp}
\ee
with probability $1-P^{iid}_W$, where
\be
P^{iid}_W=2\sum_{i=1}^m{\exp\left[-\frac{N\epsilon^{2+\frac{2c_i}{p_i}}}{M^{2+\frac{2c_i}{p_i}}m^{2+\frac{2c_i}{p_i}}K_{C_i,p_i}}\right]},
\ee
where for all $i\in\{1,\dots,m\}$, $K_{C_i,p_i}$ is a constant independent of $\bm\rho$, defined in Eq.~(\ref{KCp}).
Moreover, if the tested state is perfect, i.e., $\rho^{\otimes N+M}=\ket\psi\!\bra\psi^{\otimes N+M}$, then $F(\ket{\bm\psi}^{\otimes M},\bm\rho^{\otimes M})=1$. In that case, by Eq.~(\ref{estimateWapp}),
\be
W_{\bm\psi}\ge1-\epsilon,
\ee
with probability $1-P^{iid}_W$.

\medskip

We now generalise the previous proof to the case where the target state is of the form
\be
\ket{\bm\psi}=\hat S(\bm\xi)\hat D(\bm\beta)\,\hat U\bigotimes_{i=1}^m\ket{C_i},
\ee
where for all $i\in\{1,\dots,m\}$ each of the states $\ket{C_i}=\sum_{n=0}^{c_i-1}{c_{i,n}\ket n}$ is a core state, where $\hat U$ is an $m$-mode passive linear transformation with $m\times m$ unitary matrix $U$ and where $\bm\xi,\bm\beta\in\mathbb C^m$. To that end, we make use of the properties of heterodyne detection: by Lemma~\ref{lemma:hetmagic}, writing $\hat V=\hat S(\bm\xi)\hat D(\bm\beta)\,\hat U$, the POVM $\{\hat V\Pi^{\bm0}_{\bm\alpha}\hat V^\dag\}_{\bm\alpha\in\mathbb C^m}$ can be simulated with the POVM $\{\Pi^{\bm\xi}_{\bm\gamma}\}_{\bm\gamma\in\mathbb C^m}$ by computing $\bm\alpha=U^\dag(\bm\gamma-\bm\beta)$, i.e., translating the vector of samples $\bm\gamma$ by the vector of complex amplitudes $-\bm\beta$ and multiplying the vector obtained by the $m\times m$ unitary matrix $U^\dag$.
Formally, let $\bm\rho$ be an $m$-mode (mixed) state, then,
\be
\ba
\Tr\,(V^\dag\bm\rho\hat V\,\Pi_{\bm\alpha}^{\bm0})&=\Tr\,(\bm\rho\,\hat V\Pi_{\bm\alpha}^{\bm0}V^\dag)\\
&=\Tr\,(\bm\rho\,\Pi_{U\bm\alpha+\bm\beta}^{\bm\xi}).
\ea
\ee
In particular, computing the fidelity witness estimate $W_{\bm\psi}$ in Protocol~\ref{prot:weapp}, i.e., using samples $\bm\gamma$ from unbalanced heterodyne detection of copies of $\rho$ post-processed as $\bm\alpha=U^\dag(\bm\gamma-\beta)$, gives
\be
1-m\left[1-F\left(\bigotimes_{i=1}^m\ket{C_i}\!\bra{C_i}^{\otimes M},(V^\dag\bm\rho V)^{\otimes M}\right)\right]-\epsilon\le W_{\bm\psi}\le F\left(\bigotimes_{i=1}^m\ket{C_i}\!\bra{C_i}^{\otimes M},(V^\dag\bm\rho V)^{\otimes M}\right)+\epsilon,
\ee
with probability $1-P_W^{iid}$ by Eq.~(\ref{estimateWapp}). Given that
\be
F\left(\bigotimes_{i=1}^m\ket{C_i}\!\bra{C_i}^{\otimes M},(V^\dag\bm\rho V)^{\otimes M}\right)=F(\bm\psi^{\otimes M},\bm\rho^{\otimes M}),
\ee
we finally obtain
\be
1-m\left[1-F(\ket{\bm\psi}^{\otimes M},\bm\rho^{\otimes M})\right]-\epsilon\le W_{\bm\psi}\le F(\ket{\bm\psi}^{\otimes M},\bm\rho^{\otimes M})+\epsilon,
\label{estimateWapp2}
\ee
with probability $1-P^{iid}_W$, where $P^{iid}_W$ is defined in Eq.~(\ref{PiidW}).

\end{proof}

\noindent The estimate of the fidelity witness obtained is tight when the fidelity is close to $1$. 
Writing $n_i$ the number of indices $j\in\{1,\dots,m\}$ such that $\frac{c_j}{p_j}=\frac{c_i}{p_i}$, for all $i\in\{1,\dots,m\}$, the number of samples needed for a precision $\epsilon>0$ of the estimate $W_{\bm\psi}$ and a confidence $1-\delta$ scales as
\be
N_2=\mathcal O\left(\max_i\left\{\left(\frac{Mm}\epsilon\right)^{2+\frac{2c_i}{p_i}}\log(n_i)\log\left(\frac1\delta\right)\right\}\right),
\ee
with Eq.~(\ref{PiidW}).

%--------------------------------------------------------------------------------

\section{Removing the i.i.d.\ assumption for Protocol~\ref{prot:we}}
\label{app:wenotiid}

\noindent In this section, we derive a version of Protocol~\ref{prot:we} which does not assume i.i.d.\ state preparation. 

For $N,M\in\mathbb N^*$, let $\mathcal H_i$ and $\mathcal K_k$ denote single-mode infinite-dimensional Hilbert spaces, for all $i\in\{1,\dots,m\}$ and all $k\in\{1,\dots,N+M\}$. Let $\mathcal H$ be an $m\times(N+M)$-mode infinite-dimensional Hilbert space. The following spaces are isomorphic:
\be
\ba
\mathcal H&\simeq\mathcal K_1^{\otimes m}\otimes\dots\otimes\mathcal K_{N+M}^{\otimes m}\\
&\simeq\mathcal H_1^{\otimes(N+M)}\otimes\dots\otimes\mathcal H_m^{\otimes(N+M)}.
\ea
\ee
For any $\bm\rho^{m\times(N+M)}\in\mathcal H$, we define
\be
\bm\sigma_k^m=\Tr_{\mathcal H\setminus\mathcal K_k^{\otimes m}}(\bm\rho^{m\times(N+M)}),
\label{sigmareduced}
\ee
the $m$-mode reduced state over $\mathcal K_k^{\otimes m}$, for all $k\in\{1,\dots,N+M\}$. We also define
\be
\bm\rho_i^{N+M}=\Tr_{\mathcal H\setminus\mathcal H_i^{\otimes(N+M)}}(\bm\rho^{m\times(N+M)}),
\label{rhoreduced}
\ee
the $(N+M)$-mode reduced state over $\mathcal H_i^{\otimes(N+M)}$, for all $i\in\{1,\dots,m\}$ (see Fig.~\ref{fig:subsystems}).

The protocol is then as follows:

\begin{prot}[General multimode fidelity witness estimation]\label{prot:wenotiid}
Let $c_1,\dots,c_m\in\mathbb N^*$. Let $\ket{C_i}=\sum_{n=0}^{c_i-1}{c_{i,n}\ket n}$ be a core state, for all $i\in\{1,\dots,m\}$. Let $\hat U$ be an $m$-mode passive linear transformation with $m\times m$ unitary matrix $U$, and let $\bm\beta,\bm\xi\in\mathbb C^m$. We write $\ket{\bm\psi}=\hat S(\bm\xi)\hat D(\bm\beta)\,\hat U\bigotimes_{i=1}^m\ket{C_i}$ the $m$-mode target pure state. Let $N,M\in\mathbb N^*$, and let $p_1,\dots,p_m\in\mathbb N^*$, $0<\eta_1,\dots,\eta_m<1$, and $E_1,\dots,E_m,S_1,\dots,S_m\in\mathbb N$ be free parameters. We write $N=N'+K+Q$, for $N',K,Q\in\mathbb N$. Let $\bm\rho^{m\times(N+M)}\in\mathcal H$ be an unknown quantum state over $m\times(N+M)$ subsystems.
\begin{enumerate}
\item
Measure all the $m$ subsystems of $N=N'+K+Q$ chosen at random $m$-mode reduced states $\bm\sigma_k^m$ (see Eq.~(\ref{sigmareduced})) of $\bm\rho^{m\times(N+M)}$ with unbalanced heterodyne detection with unbalancing parameter $\bm\xi$, obtaining the vectors of samples $\bm\gamma^{(1)},\dots,\bm\gamma^{(N'+K+Q)}\in\mathbb C^m$. Let $\bm\rho^{m\times M}$ be the remaining state over $m\times M$ subsystems.
\item
Discard the last $Q$ vectors of samples.
\item
For all $k\in\{1,\dots,N'+K\}$, compute the vectors $\bm\alpha^{(k)}=U^\dag\left(\bm\gamma^{(k)}-\bm\beta\right)$. We write $\bm\alpha^{(k)}=(\alpha^{(k)}_1,\dots\alpha^{(k)}_m)\in\mathbb C^m$.
\item
For all $i\in\{1,\dots,m\}$, record the number $R_i$ of samples $\alpha_i^{(k)}$ such that $|\alpha^{(k)}_i|^2+1>E_i$, for $k\in\{N'+1,\dots,N'+K\}$. The protocol aborts if $R_i>S_i$ for any $i\in\{1,\dots,m\}$.
\item
For all $i\in\{1,\dots,m\}$, compute the mean $F_{C_i}(\bm\rho)$ of the function $z\mapsto g_{C_i}^{(p_i)}(z,\eta_i)$ (defined in Eq.~(\ref{gCmain})) over the samples $\alpha_i^{(1)},\dots,\alpha_i^{(N')}\in\mathbb C$.
\item
Compute the fidelity witness estimate $W_{\bm\psi}=1-\sum_{i=1}^m{(1-F_{C_i}(\bm\rho)^M)}$.
\end{enumerate}
\end{prot}

\noindent Akin to Protocol~\ref{prot:fenotiid} compared to Protocol~\ref{prot:fe}, this protocol differs from Protocol~\ref{prot:we} by two additional classical post-processing steps, steps 2 and 4. These steps are part of a de Finetti reduction for infinite-dimensional systems~\cite{renner2009finetti} detailed in~\cite{chabaud2019buildingarXiv}. In particular, step 3 is an energy test, and the energy parameters $E_i$ and $S_i$ should be chosen to guarantee completeness, i.e., if the perfect state is sent then it passes the energy test with high probability. 

\begin{figure}
	\begin{center}
		\includegraphics[width=0.6\columnwidth]{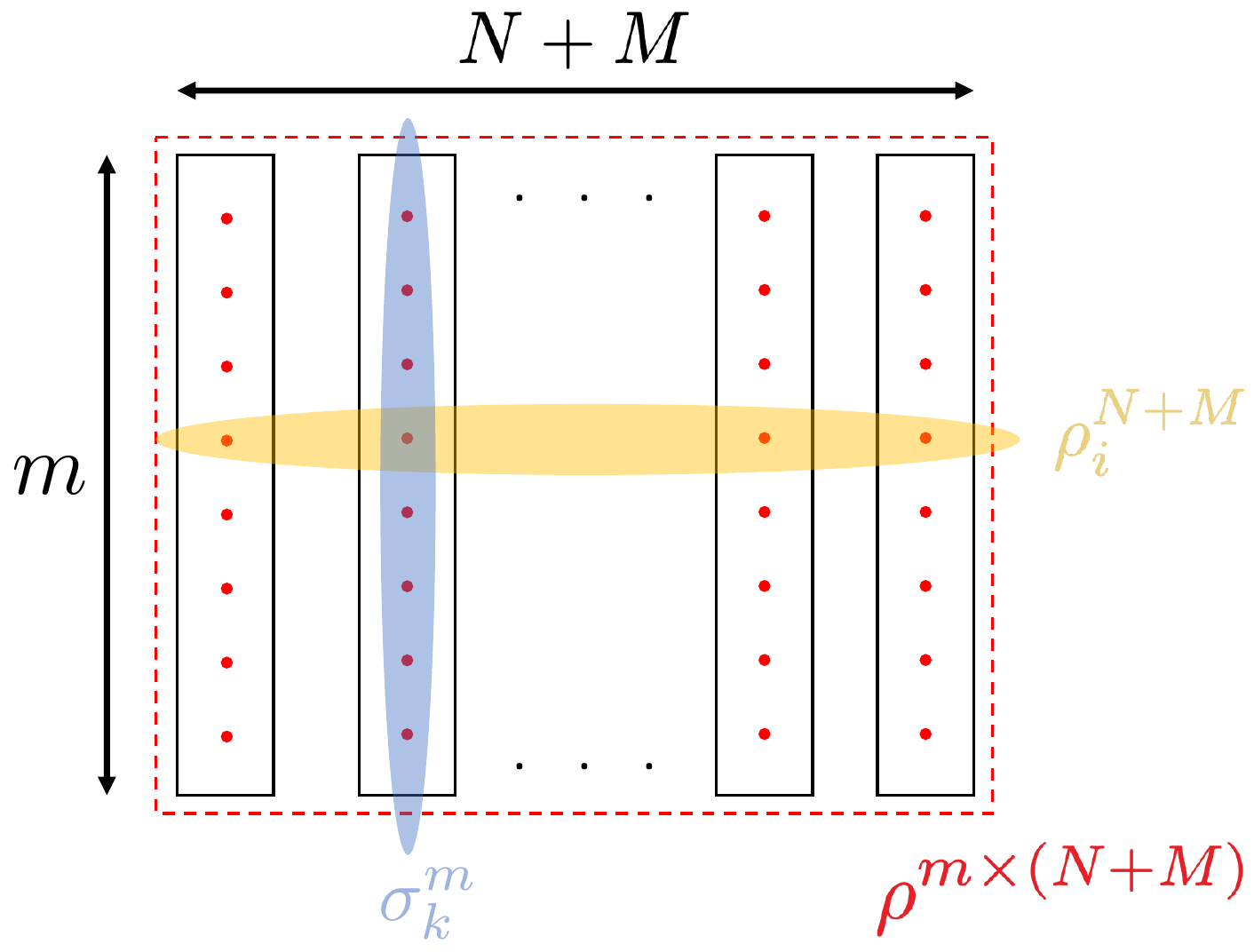}
		\caption{The notations for the subsystems of a quantum state $\bm\rho^{m\times(N+M)}$ over $m\times(N+M)$ modes, seen as $N+M$ groups of $m$ subsystems or as $m$ groups of $N+M$ subsystems. The red dots depict the single-mode subsystems of $\bm\rho^{m\times(N+M)}$. The reduced state of $\bm\rho^{m\times(N+M)}$ corresponding to the $i^{th}$ group of $N+M$ subsystems is denoted $\bm\rho_i^{N+M}$, for all $i\in\{1,\dots,m\}$, and the reduced state of $\bm\rho^{m\times(N+M)}$ corresponding to the $k^{th}$ group of $m$ subsystems is denoted $\bm\sigma_k^m$, for all $k\in\{1,\dots,N+M\}$.}
		\label{fig:subsystems}
	\end{center}
\end{figure}

The value $W_{\bm\psi}$ obtained is an estimate of a tight lower bound of the fidelity between the remaining state $\rho^{m\times M}$ over $m\times M$ subsystems and $M$ copies of the $m$-mode target state $\ket{\bm\psi}$. The efficiency of the protocol is summarised by the following result:

\begin{theo}\label{th:wenotiid}
Let $\epsilon>0$. With the notations of Protocol~\ref{prot:fenotiid} and $\eta\le\eta_{C,p}$, defined in Eq.~(\ref{etafeapp}),
\be
1-m\left[1-F(\ket{\bm\psi}^{\otimes M},\bm\rho^{m\times M})\right]-\mathcal O(\epsilon)\le W_{\bm\psi}\le F(\ket{\bm\psi}^{\otimes M},\bm\rho^{m\times M})+\mathcal O(\epsilon)
\ee
or the protocol aborts in step 3, with probability greater than $1-P_W$,
where
\be
P_W=\mathcal O\left(\frac1{\poly(m,M,\frac1\epsilon)}\right),
\ee
for
\be
N=\mathcal O\left(\frac{M^{7+\frac{2c}p}m^{4+\frac{4c}p}}{\epsilon^{4+\frac{4c}p}}\right).
\ee
\end{theo}

\begin{proof}

For the proof we only consider the case where the target state is a tensor product of core states $\ket{\bm\psi}=\bigotimes_{i=1}^m{\ket{C_i}}$, i.e., $\hat S(\bm\xi)\hat D(\bm\beta)\,\hat U=\mathbb1$, since the generalisation to the multimode target states of the form $\hat S(\bm\xi)\hat D(\bm\beta)\,\hat U\bigotimes_{i=1}^m\ket{C_i}$ using Lemma~\ref{lem:hetmagic} is identical to the i.i.d.\ case, detailed in section~\ref{app:thwe}. Then, the proof is similar to the one of Theorem~\ref{th:weapp}, using Theorem~\ref{th:fenotiid} instead of Theorem~\ref{th:feapp}.

Writing $\bm\rho^{m\times M}\in\mathcal H_1^{\otimes M}\otimes\dots\otimes\mathcal H_m^{\otimes M}$, for all $i\in\{1,\dots,m\}$ (see Fig.~\ref{fig:subsystems} with $N=0$), let
\be
\bm\rho_i^M=\Tr_{\bigotimes_{j\neq i}\mathcal H_j^{\otimes M}}(\bm\rho^{m\times M}),
\ee
be the reduced state over $\mathcal H_i^{\otimes M}$ of the state $\bm\rho^{m\times M}$.
By Theorem~\ref{th:fenotiid} and Eq.~(\ref{estimateFnotiidbis}) in particular,
\be
\left|F(\ket{C_i}^{\otimes M},\bm\rho_i^M)-F_{C_i}(\bm\rho)^M\right|\le3\epsilon
\ee
for $M$ large enough, or the protocol aborts in step 3, with probability greater than $1-P_{C_i}$,
where
\be
P_{C_i}=\mathcal O\left(\frac1{\poly(M,\frac1\epsilon)}\right),
\ee
for
\be
N=\mathcal O\left(\frac{M^{7+\frac{4c}p}}{\epsilon^{4+\frac{4c}p}}\right).
\ee
Replacing $\epsilon$ by $\frac\epsilon m$ and taking the union bound of the failure probabilities we obtain
\be
\ba
\left|1-\sum_{i=1}^m{\left(1-F(\ket{C_i}^{\otimes M},\bm\rho_i^M)\right)}-W_{\bm\psi}\right|&=\left|\sum_{i=1}^m{\left(F(\ket{C_i}^{\otimes M},\bm\rho_i^M)-F_{C_i}(\bm\rho)^M\right)}\right|\\
&\le\sum_{i=1}^m{\left|F(\ket{C_i}^{\otimes M},\bm\rho_i^M)-F_{C_i}(\bm\rho)^M\right|}\\
&=\mathcal O(\epsilon),
\label{finalO}
\ea
\ee
or the protocol aborts in step 3, with probability greater than $1-P_W$,
where
\be
P_W=\mathcal O\left(\frac1{\poly(m,M,\frac1\epsilon)}\right),
\ee
for
\be
N=\mathcal O\left(\frac{M^{7+\frac{2c}p}m^{4+\frac{2c}p}}{\epsilon^{4+\frac{2c}p}}\right).
\ee
By Lemma~\ref{lem:productmulti} we have
\be
1-m\left[1-F(\ket{\bm\psi}^{\otimes M},\bm\rho^{m\times M})\right]\le1-\sum_{i=1}^m{\left(1-F(\ket{C_i}^{\otimes M},\bm\rho_i^M)\right)}\le F(\ket{\bm\psi}^{\otimes M},\bm\rho^{m\times M}),
\ee
so with Eq.~(\ref{finalO}) we finally obtain
\be
1-m\left[1-F(\ket{\bm\psi}^{\otimes M},\bm\rho^{m\times M})\right]-\mathcal O(\epsilon)\le W_{\bm\psi}\le F(\ket{\bm\psi}^{\otimes M},\bm\rho^{m\times M})+\mathcal O(\epsilon)
\ee
or the protocol aborts in step 3, with probability greater than $1-P_W$,
where
\be
P_W=\mathcal O\left(\frac1{\poly(m,M,\frac1\epsilon)}\right),
\ee
for
\be
N=\mathcal O\left(\frac{M^{7+\frac{4c}p}m^{4+\frac{4c}p}}{\epsilon^{4+\frac{4c}p}}\right).
\ee

\end{proof}

\noindent Since the parameter $p\in\mathbb N^*$ may vary freely, the asymptotic scaling of the number of samples $N$ needed for a precision $\epsilon>0$ and a polynomial confidence thus is given by $N=\mathcal O(\frac{M^7m^4}{\epsilon^4})$.

%--------------------------------------------------------------------------------

\section{Optimised verification of Boson Sampling}
\label{app:optiBS}

\noindent We consider the case $M=1$, the general case being retrieved by setting $\epsilon=\frac\epsilon M$. 

Using Protocol~\ref{prot:we} directly to certify the output states of a Boson Sampling experiment yields an efficiency $N_2=\mathcal O(m^{2+\frac2p}\log m\log\frac1\delta)$ for a constant error in the estimation, with a failure probability $\delta$, where the constant prefactor scales with the free parameter $p$, giving an asymptotic scaling of $\mathcal O(m^2\log m\log\frac1\delta)$. We show in this section that a refined version of Protocol~\ref{prot:we} using a single-mode fidelity witness as a subroutine rather than a single-mode fidelity estimate (from Protocol~\ref{prot:feapp}) actually yields an efficiency $N_3=\mathcal O(m^2\log m\log\frac1\delta)$ already in the finite regime, thus proving Theorem~\ref{th:bs}.

We use the fact that the target core states appearing in the expressions of Boson Sampling output states are Fock states, in order to control the analytical error from Lemma~\ref{lem:inductiong} derived in section~\ref{app:lemmain}, which we recall below.

Let $p\in\mathbb N^*$, let $k,l\in\mathbb N$ and let $0<\eta<1$. Let $\rho=\sum_{i,j=0}^{+\infty}{\rho_{ij}\ket i\!\bra j}$ be a density operator.
Then,
\be
\underset{\alpha\leftarrow Q_\rho}{\mathbb E}[g_{k,l}^{(p)}(\alpha,\eta)]=\rho_{kl}+(-1)^{p+1}\sum_{q=p}^{+\infty}{\rho_{k+q,l+q}\eta^q\binom{q-1}{p-1}\sqrt{\binom{k+q}k\binom{l+q}l}},
\ee
where the function $g_{k,l}$ is defined in Eq.~(\ref{appg}). In particular, when $k=l$, for a target Fock state $\ket k\!\bra k$, and for $p$ even we obtain
\be
\underset{\alpha\leftarrow Q_\rho}{\mathbb E}[g_{k,k}^{(p)}(\alpha,\eta)]=\rho_{kk}-\sum_{q=p}^{+\infty}{\rho_{k+q,k+q}\eta^q\binom{q-1}{p-1}\binom{k+q}k}.
\ee
The sum on the right hand side is positive, so $\underset{\alpha\leftarrow Q_\rho}{\mathbb E}[g_{k,k}^{(p)}(\alpha,\eta)]$ is always a lower bound on the fidelity between the state $\rho$ and the Fock state $\ket k\!\bra k$ (with analytical error bounded independently of $\rho$ by Eq.~(\ref{betterboundEgkl})). Moreover, this lower bound is tight when $\rho$ is close to the Fock state $\ket k\!\bra k$. 

With Eqs.~(\ref{appgC}), (\ref{boundgC}) and (\ref{BCp}) we obtain the bounds
\begin{equation}
|g_{0,0}^{(p)}|\le\frac p\eta,
\end{equation}
for the Fock state $\ket0$,  and
\begin{equation}
|g_{1,1}^{(p)}|\le\frac{p(p+1)}{2\eta^2},
\end{equation}
for the Fock state $\ket1$.  Let $\epsilon>0$,  and let $\alpha_1,\dots,\alpha_N$ be $N$ samples from heterodyne detection of a state $\rho$ we have by Lemma~\ref{lem:Hoeffding} (Hoeffding inequality):
\begin{equation}
\Pr\left[\left|\frac1N\sum_{i=1}^N{g_{0,0}^{(p)}(\alpha_i,\eta)}-\underset{\alpha\leftarrow Q_\rho}{\mathbb E}[g_{0,0}^{(p)}(\alpha,\eta)]\right|\ge\frac\epsilon m\right] \le2\exp\left[{-\frac{N\epsilon^2\eta^2}{2p^2m^2}}\right],
\end{equation}
and
\begin{equation}
\Pr\left[\left|\frac1N\sum_{i=1}^N{g_{1,1}^{(p)}(\alpha_i,\eta)}-\underset{\alpha\leftarrow Q_\rho}{\mathbb E}[g_{1,1}^{(p)}(\alpha,\eta)]\right|\ge\frac\epsilon m\right]\le2\exp\left[{-\frac{2N\epsilon^2\eta^4}{p^2(p+1)^2m^2}}\right].
\end{equation}
In the case of Boson Sampling, the tensor product input is given by $n$ Fock states $\ket1$ and $m-n$ Fock states $\ket0$. With the union bound of the failure probabilies above,  using these single-mode fidelity witnesses in Protocol~\ref{prot:we} instead of the fidelity estimate from Protocol~\ref{prot:feapp}, we obtain an estimate with precision $\epsilon$ of a tight multimode fidelity witness with failure probability
\begin{equation}
P_{BS}^{iid}=2(m-n)\exp\left[{-\frac{N\epsilon^2\eta^2}{2p^2m^2}}\right]+2n\exp\left[{-\frac{2N\epsilon^2\eta^4}{p^2(p+1)^2m^2}}\right],
\end{equation}
when using $N$ samples from tensor product single-mode heterodyne detection. In particular, for a constant precision and a failure probability $\delta$,  it is sufficient to use $N_3=\mathcal O(m^2\log m\log\frac1\delta)$ samples from product single-mode heterodyne detection. The constant prefactor may still be optimised by playing with the choice of $p$ even and $0<\eta<1$.

Note that this optimised protocol can also be employed in the more general case of the certification of multimode quantum states of the form
\be
\left(\bigotimes_{i=1}^m\hat G_i\right)\hat U\left(\bigotimes_{i=1}^m\ket{n_i}\right),
\ee
where each state $\ket{n_i}$ is a single-mode Fock state, where $\hat U$ is an $m$-mode passive linear transformation and where each $\hat G_i$ is a single-mode Gaussian unitary operation, for all $i\in\{1,\dots,m\}$.

%--------------------------------------------------------------------------------

\end{document}